%% file: Submission-v3.tex
\documentclass[aps,%
twocolumn,%
floatfix,%
pra,%
amsfonts,%
groupedaddress,%
superscriptaddress]{revtex4-2}%
\usepackage[utf8]{inputenc}
\usepackage[T1]{fontenc}
\usepackage{mathrsfs}       
\usepackage{mathtools}
\usepackage{amsmath}
\usepackage{amsfonts}       
\usepackage[titletoc,title]{appendix}
\usepackage{algorithm,algorithmic}

\usepackage{bbm}            
\usepackage{bm}             
\usepackage{amsthm}         
\usepackage{makecell}
\usepackage[dvipsnames,table,xcdraw]{xcolor}
\definecolor{beamer@blendedblue}{rgb}{0.2,0.2,0.7}
\usepackage{times}
\usepackage{caption}
\usepackage{subfig}
\usepackage{overpic}
\usepackage{booktabs} 
\usepackage{multirow} 
\usepackage{wasysym}

\newcolumntype{C}[1]{>{\centering\arraybackslash}p{#1}}

\usepackage[colorlinks=true,%
linkcolor=beamer@blendedblue,%
bookmarks=true,%
breaklinks=true,%
filecolor=beamer@blendedblue,%
anchorcolor=yellow,%
citecolor=beamer@blendedblue,%
urlcolor=beamer@blendedblue,%
pdfauthor={Congcong Zheng, Kun Wang, Xutao Yu, Ping Xu, and Zaichen Zhang},%
pdfsubject={Distributed Quantum Inner Product Estimation with Structured Random Circuits},%
CJKbookmarks=true]{hyperref}
\usepackage[capitalise,nameinlink]{cleveref} 

\usepackage{float}    
\usepackage{pifont}

\newtheorem{definition}{Definition}

\newtheorem{lemma}[definition]{Lemma}
\newtheorem{theorem}[definition]{Theorem}

\newtheorem{example}[definition]{Example}

\mathchardef\ordinarycolon\mathcode`\:
\mathcode`\:=\string"8000
\def\vcentcolon{\mathrel{\mathop\ordinarycolon}}
\begingroup \catcode`\:=\active
  \lowercase{\endgroup
  \let :\vcentcolon
  }

\DeclareFontFamily{U}{mathx}{\hyphenchar\font45}
\DeclareFontShape{U}{mathx}{m}{n}{<-> mathx10}{}
\DeclareSymbolFont{mathx}{U}{mathx}{m}{n}
\DeclareMathAccent{\widebar}{0}{mathx}{"73}

\newcommand{\ket}[1]{\vert{#1}\rangle}
\newcommand{\bra}[1]{\langle{#1}\vert}
\newcommand{\ketbra}[1]{\vert{#1}\rangle\!\langle{#1}\vert}

\DeclareMathOperator{\tr}{Tr}  
\newcommand{\1}{\mathbbm{1}}
\newcommand{\ox}{\otimes}

\newcommand{\norm}[2]{\ensuremath{\left\lVert#1\right\rVert_{#2}}}%

\newcommand{\Var}{\mathbb{V}}
\newcommand{\Cl}{{\rm Cl}}
\newcommand{\ba}{\bm{a}}
\newcommand{\bb}{\bm{b}}
\newcommand{\bx}{\bm{x}}

\makeatletter
\newsavebox{\@brx}
\newcommand{\llangle}[1][]{\savebox{\@brx}{\(\m@th{#1\langle}\)}%
  \mathopen{\copy\@brx\kern-0.5\wd\@brx\usebox{\@brx}}}
\newcommand{\rrangle}[1][]{\savebox{\@brx}{\(\m@th{#1\rangle}\)}%
  \mathclose{\copy\@brx\kern-0.5\wd\@brx\usebox{\@brx}}}
\makeatother

\newcommand*{\cA}{\mathcal{A}}
\newcommand*{\cB}{\mathcal{B}}

\newcommand*{\cD}{\mathcal{D}}
\newcommand*{\cE}{\mathcal{E}}
\newcommand*{\cF}{\mathcal{F}}

\newcommand*{\cH}{\mathcal{H}}

\newcommand*{\cM}{\mathcal{M}}

\newcommand*{\cO}{\mathcal{O}}
\newcommand*{\cP}{\mathcal{P}}

\newcommand*{\cR}{\mathcal{R}}
\newcommand*{\cS}{\mathcal{S}}
\newcommand*{\cT}{\mathcal{T}}
\newcommand*{\cU}{\mathcal{U}}

\newcommand*{\cZ}{\mathcal{Z}}

\newcommand{\bE}{\mathbb{E}}

\newcommand{\bZ}{\mathbb{Z}}
\newcommand{\bS}{\mathbb{S}}
\newcommand{\bR}{\mathbb{R}}
\newcommand{\bP}{\mathbb{P}}
\newcommand{\bF}{\mathbb{F}}
\newcommand{\bX}{\mathbb{X}}
\newcommand{\bY}{\mathbb{Y}}

\definecolor{wildstrawberry}{rgb}{1.0, 0.26, 0.64}
\definecolor{googleblue}{HTML}{4285F4}
\definecolor{googlered}{HTML}{DB4437}
\definecolor{googleyellow}{HTML}{F4B400}
\definecolor{googlegreen}{HTML}{0F9D58}
\definecolor{klevinblue}{HTML}{002FA7}
\definecolor{tiffanyblue}{HTML}{0ABAB5}

\begin{document}


\newcommand{\thetitle}{{Distributed Quantum Inner Product Estimation with Structured Random Circuits}}
\title{\thetitle}
\author{Congcong Zheng}%
\affiliation{State Key Lab of Millimeter Waves, Southeast University, Nanjing 211189, China}%
\affiliation{Purple Mountain Laboratories, Nanjing 211111, China}%
\affiliation{Frontiers Science Center for Mobile Information Communication and Security, Southeast University, Nanjing 210096, China}%

\author{Kun Wang}
\thanks{Corresponding author: \href{nju.wangkun@gmail.com}{nju.wangkun@gmail.com}}%
\affiliation{College of Computer Science and Technology, National University of Defense Technology, Changsha 410073, China}%

\author{Xutao Yu}%
\affiliation{State Key Lab of Millimeter Waves, Southeast University, Nanjing 211189, China}%
\affiliation{Purple Mountain Laboratories, Nanjing 211111, China}%
\affiliation{Frontiers Science Center for Mobile Information Communication and Security, Southeast University, Nanjing 210096, China}%

\author{Ping Xu}%
\affiliation{College of Computer Science and Technology, National University of Defense Technology, Changsha 410073, China}%

\author{Zaichen Zhang}%
\thanks{Corresponding author: \href{zczhang@seu.edu.cn}{zczhang@seu.edu.cn}}%
\affiliation{National Mobile Communications Research Laboratory, Southeast University, Nanjing 210096, China}%
\affiliation{Purple Mountain Laboratories, Nanjing 211111, China}%
\affiliation{Frontiers Science Center for Mobile Information Communication and Security, Southeast University, Nanjing 210096, China}%

\begin{abstract}
Distributed inner product estimation (DIPE) is a fundamental task in quantum information, 
aiming to estimate the inner product between two unknown quantum states prepared on distributed quantum platforms.
Existing rigorous sample complexity analyses are limited to unitary $4$-designs, which pose significant practical challenges for near-term quantum devices.
This work addresses this challenge by exploring DIPE with structured random circuits. 
We first establish that DIPE with an arbitrary unitary $2$-design ensemble achieves an average sample complexity of $\cO(\sqrt{2^n})$, where $n$ is the number of qubits.
We then analyze ensembles below unitary $2$-designs---specifically, the brickwork and local unitary $2$-design ensembles---demonstrating average sample complexities of $\mathcal{O}(\sqrt{2.18^n})$ and $\mathcal{O}(\sqrt{2.5^n})$, respectively. 
Furthermore, we analyze the state-dependent sample complexity. 
For brickwork ensembles, we develop a tensor network approach to compute the asymptotic state-dependent sample complexity, showing that it converges to $\cO(\sqrt{2.18^n})$ as the circuit depth increases.
Remarkably, we find that DIPE with the global Clifford ensemble requires $\Theta(\sqrt{2^n})$ copies, matching the performance of unitary $4$-designs.
For both local and global Clifford ensembles, we find that the efficiency can be further enhanced by the nonstabilizerness of states. 
Additionally, for approximate unitary $4$-designs, the performance exponentially approaches that of exact unitary $4$-designs as the circuit depth increases.
Our results provide theoretically guaranteed methods for implementing DIPE with experimentally feasible unitary ensembles.
\end{abstract}
\date{\today}
\maketitle

\section{Introduction}
The engineering and physical realization of quantum computers and quantum simulators are being actively pursued across various physical platforms~\cite{popkin2016quest, preskill2018quantum, brown20245}.
To certify their performance, numerous protocols have been developed to compare experimentally generated quantum states or processes against known theoretical targets,  
including direct fidelity estimation~\cite{flammia2011direct, dasilva2011practical, leone2023nonstabilizerness}, 
random benchmarking~\cite{emerson2007symmetrized, lu2015experimental, helsen2022general}, and 
quantum verification~\cite{pallister2018optimala, wang2019optimala, zheng2024efficient, chen2025quantum}.
However, a significant challenge remains: how to directly compare unknown quantum states (or processes) generated on different physical platforms, at different locations and times.
This task, known as \emph{cross-platform verification}, becomes especially relevant as we enter 
the quantum advantage regime where classical simulation of quantum systems becomes computationally intractable.

To address this challenge, Elben \emph{et al.} proposed the first cross-platform protocol for estimating the similarity 
between two unknown quantum states prepared on distant quantum platforms~\cite{elben2020crossplatform}. 
Subsequently, Zhu \emph{et al.} reported the first experimental demonstration of cross-platform verification across different quantum computing platforms~\cite{zhu2022crossplatforma}.
Extensions to quantum processes have been proposed in~\cite{knorzer2023crossplatform, zheng2024crossplatform}. 
Recent efforts have aimed to enhance the efficiency of cross-platform verification through various techniques, including Pauli sampling~\cite{hinsche2025efficient}, Bell sampling~\cite{denzler2025highlyentangled}, deep learning~\cite{qian2024multimodal}, and quantum links~\cite{knorzer2023crossplatform, gong2024sample, arunachalam2024distributed}. 
At the heart of cross-platform verification lies the task of \emph{distributed inner product estimation (DIPE)}. 
A key theoretical advance was made by Anshu \textit{et al.}, 
who proved that DIPE with a unitary $4$-design ensemble requires $\Theta(\sqrt{2^n})$ state copies 
for two $n$-qubit quantum states in the worst case~\cite{anshu2022distributed}.

However, a significant obstacle hinders the practical implementation of Anshu's protocol: 
the deep circuits required for exact unitary $4$-designs far exceed the capabilities of near-term quantum devices, 
primarily due to circuit depth limitations~\cite{schuster2025random}.
For instance, current quantum platforms exhibit typical noise rates of $\alpha = 0.5\%$, 
allowing roughly $1/\alpha \approx 200$ reliable gate operations. 
Circuits of $\mathcal{O}(n)$ depth would restrict DIPE to fewer than $15$ qubits ($n^2 \approx 200$).
This crucial limitation impedes the immediate application of these powerful theoretical results and risks delaying the real-world impact of DIPE.
Recognizing this limitation, exploring DIPE with more experimentally feasible unitary ensembles is both vital and urgent. 
Specifically, the following important questions remain largely open:
(i) What is the sample complexity of DIPE with the widely studied Clifford ensemble?
(ii) Can DIPE be efficiently performed with low-depth circuit ensembles?
Notably, the first question was also raised in~\cite{chen2024nonstabilizerness}, and low-depth circuit ensembles---being easier to implement 
than exact unitary $4$-designs---have attracted considerable attention in recent quantum information research~\cite{bravyi2020quantum, bravyi2018quantum,  huang2024learning, malz2024preparation, landau2024learning, yang2025compression, schuster2025random, bertoni2024shallow, ippoliti2023operator, akhtar2023scalable, hu2023classical, bu2024classical, hu2025demonstration}.

\begin{figure*}[!htbp]
\centering
\includegraphics[width=\linewidth]{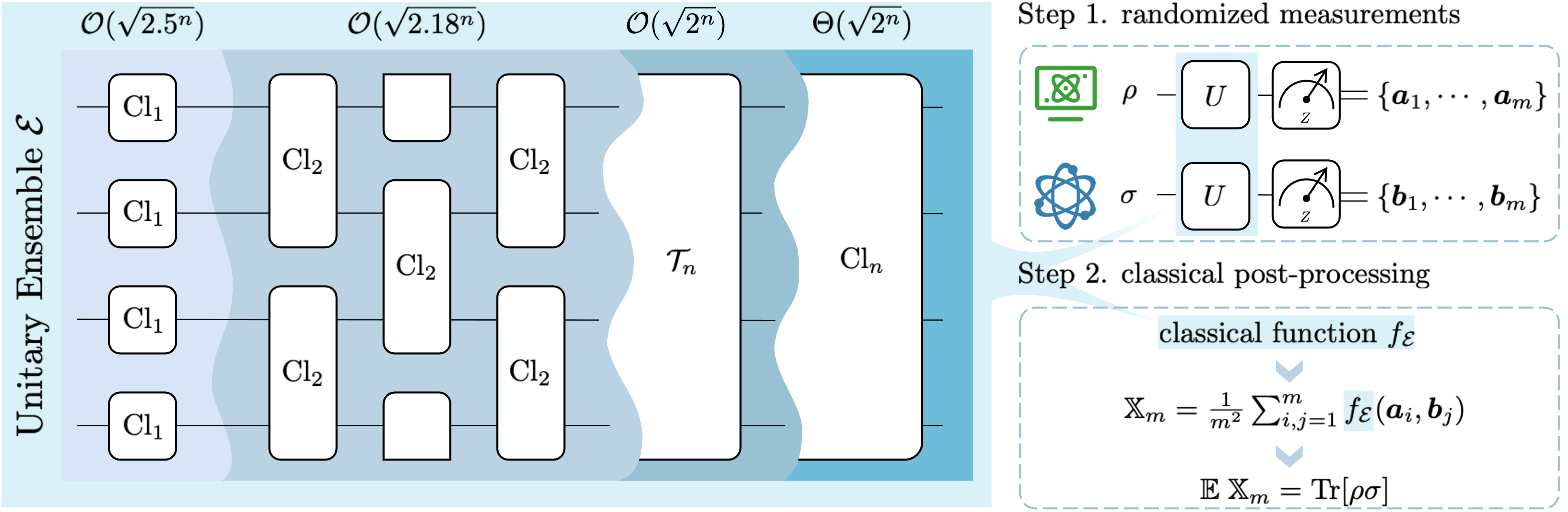}
\caption{\raggedright
The general framework for distributed inner product estimation (DIPE).
Here, $\rho$ and $\sigma$ are two quantum states independently prepared on two distant quantum platforms.
The DIPE begins by applying randomized measurements on each platform using a unitary ensemble $\cE$. 
The resulting measurement outcomes are then processed classically using a function $f_{\cE}$, 
which depends on $\cE$, to obtain an unbiased estimator of the inner product $\tr[\rho\sigma]$. 
In this work, we focus on the following experimentally feasible unitary ensembles: 
(i) the $n$-qubit global Clifford ensemble $\Cl_n$, 
(ii) the $n$-qubit unitary $2$-design ensemble $\cT_n$, 
(iii) the brickwork ensemble $\cB_d$, where $d$ denotes the depth, and 
(iv) the local Clifford ensemble $\Cl_1^{\ox n}$. 
The average sample complexities for each ensemble are shown above, where the worst-case sample complexity for ${\rm Cl}_n$ is $\Theta(\sqrt{2^n})$.
}
\label{fig: framework}
\end{figure*}

In this work, we address both of these questions. 
We present a general framework for DIPE and analyze the sample complexity with various structured random unitary ensembles. 
First, we focus on the \emph{average sample complexity}, demonstrating that DIPE is exponentially hard for most states. 
Concretely, we show that the average sample complexity of DIPE with an arbitrary unitary $2$-design is $\cO(\sqrt{2^n})$.
We then investigate unitary ensembles below unitary $2$-designs.
For local unitary $2$-designs, we show that DIPE requires $\mathcal{O}(\sqrt{2.5^n})$ state copies on average.
We also consider a representative structured random circuit ensemble: the brickwork ensemble, which has been widely employed in classical shadows~\cite{huang2020predicting, bertoni2024shallow, ippoliti2023operator, rozon2024optimala, farias2025robust, arienzo2023closedforma, schuster2025random, hu2025demonstration}. 
We demonstrate that DIPE with the brickwork ensemble requires $\mathcal{O}(\sqrt{2.18^n})$ state copies on average, notably \emph{independent} of circuit depth.

Second, to further explore the performance of DIPE, we analyze the \emph{state-dependent sample complexity} for the brickwork and Clifford ensembles.
For the brickwork ensemble, we develop a tensor network approach to compute the asymptotic state-dependent sample complexity, showing that it converges to $\cO(\sqrt{2.18^n})$ for all state pairs as the depth increases.
Remarkably, we find that DIPE with the global Clifford ensemble requires $\Theta(\sqrt{2^n})$ state copies for all states, matching the performance of unitary $4$-designs while being significantly more practical to implement. 
In contrast, DIPE with the local Clifford ensemble requires $\mathcal{O}(\sqrt{4.5^n})$ copies for stabilizer product states.
Moreover, for both local and global Clifford ensembles, we show that the nonstabilizerness of states further enhances the efficiency of DIPE.
Furthermore, we analyze the performance of DIPE with approximate unitary $4$-designs, showing that it exponentially approaches that of exact unitary $4$-designs as the circuit depth increases.
Finally, we perform numerical simulations on systems of up to $26$ qubits to validate our theoretical results.

The remaining parts of this paper are organized as follows.
In Section~\ref{sec:general framework for DIPE}, we present the general framework for DIPE. 
In Section~\ref{sec:cpe with brickwork circuits}, we analyze the average and state-dependent sample complexities of DIPE with the brickwork ensemble.
In Section~\ref{sec:cpe with clifford}, we analyze the sample complexities of DIPE with the global and local Clifford ensembles.
In Section~\ref{sec:cpe with approximate 4-design}, we discuss DIPE with approximate unitary $4$-designs.
In Section~\ref{sec:numerical simulation}, we present numerical simulations to validate our theoretical results.

\section{General Framework for DIPE}
\label{sec:general framework for DIPE}

First, we present the general framework for DIPE, as illustrated in Fig.~\ref{fig: framework}. 
In this work, we focus on $n$-qubit quantum systems with Hilbert space $\cH_n$. 
Consider two $n$-qubit platforms, each preparing an unknown quantum state, $\rho$ and $\sigma$, respectively. 
DIPE aims to estimate the inner product of these two states, $\tr[\rho\sigma]$. 

\subsection{Protocol}
Let $\cE=(\cU, \mu)$ be a unitary ensemble, where $\cU$ is a subset of the $n$-qubit unitary group and $\mu$ is a probability measure over $\cU$.
DIPE consists of two main steps~\cite{elben2020crossplatform, anshu2022distributed}.

\textbf{Step 1.} Randomized Measurements:
Randomly sample a unitary $U \sim \cE$ according to $\mu$, 
apply $U$ to both states $\rho$ and $\sigma$, and 
perform measurements in the computational basis with $m$ shots for each state. 
This yields measurement outcomes 
$\{\bm{a}_{i}\}_{i=1}^{m}$ and $\{\bm{b}_{i}\}_{i=1}^{m}$, 
where $\bm{a}_{i}, \bm{b}_{i}\in\bZ_2^n$. 

\textbf{Step 2.} Classical Post-processing:
Define a random variable
\begin{align}
    \bX_m := \frac{1}{m^2} \sum 
    f_{\cE}(\bm{a}_i, \bm{b}_j), 
    \label{eq:random variable} 
\end{align}
where $f_{\cE}:\bZ_2^n\times\bZ_2^n\to\bR$ is a \emph{classical function} that depends on the ensemble $\cE$, which will be discussed in detail later.
 
We repeat the above two steps $N$ times to obtain a collection of random variables $\{\bX_m^{(t)}\}_{t=1}^N$ and compute the \emph{mean estimator}:
\begin{align}\label{eq:estimator}
\hat{\omega} := \frac{1}{N}\sum_{t=1}^N \bX_m^{(t)},
\end{align}
which serves as an unbiased estimator of the inner product $\tr[\rho\sigma]$. 
The total number of state copies required on each device is $Nm$, 
which determines the sample complexity of the protocol.
A summary of the protocol is provided in Algorithm~\ref{alg:general framework}. 

\subsection{Classical Function}
We now discuss the choice of the classical function in Eq.~\eqref{eq:random variable}. 
The only requirement for this function is that the estimator $\bX_m$ remains unbiased.
Clearly, the choice of the classical function depends critically on the random unitary ensemble $\cE$. 
To guide this selection, we define the \emph{$k$-moment channel} of $\cE$ as 
\begin{align}
    \cM_\cE^{(k)}(A) := \bE_{U\sim \cE} U^{\dagger\ox k} A U^{\ox k}, 
    \label{eq:k-moment channel}
\end{align}
leading to the following lemma.

\begin{lemma}
\label{lem:requirement for classical function}
To guarantee that $\hat{\omega}$ defined in Eq.~\eqref{eq:estimator} is an unbiased estimator,
the classical function $f_{\cE}$ should satisfy
\begin{align}
    \tr[\cM_\cE^{(2)}(O)(P\ox P')] = 
    \begin{cases}
        2^n, & P = P', \\
        0, &\text{otherwise}.
    \end{cases},
\end{align}
for all $P,P'\in\cP_n$, 
where $\cP_n := \{I,X,Y,Z\}^{\ox n}$ is the $n$-qubit Pauli set, and 
$O := \sum_{\bm{a}, \bm{b}} f_{\cE}(\bm{a},\bm{b}) \ketbra{\bm{a}\bm{b}}$.
\end{lemma}
The proof is provided in~\cite{elben2019statistical} and Appendix~\ref{app:general framework for DIPE}. 
In~\cite{elben2020crossplatform}, the authors introduced two examples of classical functions: 
\begin{enumerate}
\item For $\cE = \cT_1^{\ox n}$, where $\cT_1$ is a unitary $2$-design on $\cH_1$, the classical function is 
$f_{\cT_1^{\ox n}}(\bm{a}, \bm{b}) = 2^n(-2)^{-\cD(\bm{a}, \bm{b})}$, where $\cD(\bm{a}, \bm{b})$ is the Hamming distance between $\bm{a}$ and $\bm{b}$;
\item For $\cE = \cT_n$, a unitary $2$-design on $\cH_n$, 
the classical function is $f_{\cT_n}(\bm{a}, \bm{b}) = 2^{n}$ if $\bm{a}=\bm{b}$ otherwise $-1$. 
\end{enumerate}

However, for other types of unitary ensembles, explicit constructions of classical functions remain largely unexplored. 
In the following, we focus on a particularly structured class known as \emph{Pauli-invariant ensembles} 
and investigate the properties of their associated classical functions.
It is worth noting that all unitary ensembles explored in this work are Pauli-invariant.
As the name suggests, an ensemble $\cE$ is Pauli-invariant if, for every unitary $U \in \cE$ and all Pauli operators $P \in \cP_n$,
both $PU$ and $UP$ are also in the ensemble with the same probability distribution~\cite{bu2024classical}.
For this kind of ensemble, we have the following lemma. 
\begin{lemma}
\label{lem:classical function requirement of Pauli-invariant ensemble}
If $\cE$ is a Pauli-invariant ensemble, 
then the corresponding classical function $f_\cE$ must satisfy
\begin{align}
f_\cE(\ba,\bb) = f_\cE(\ba\oplus\bb, \bm{0}), 
\end{align} 
where $(\ba\oplus\bb)_i = 0$ if $\ba_i=\bb_i$ and $1$ otherwise. 
\end{lemma}
See proof in Appendix~\ref{app:general framework for DIPE}. 
Hence, there are only $2^n$ distinct values that $f_\cE$ can take if $\cE$ is a Pauli-invariant ensemble.

\subsection{Sample Complexity}
We now analyze the number of state copies required to estimate $\tr[\rho\sigma]$ up to a fixed additive error $\varepsilon$ and failure probability $\delta$. 
By Chebyshev's inequality, the estimator $\hat{\omega}$ satisfies
\begin{align}
    \Pr\left\{\left\vert\hat{\omega} - \tr[\rho\sigma]\right\vert \geq \varepsilon\right\} 
    \leq \frac{\Var_{\cE}(\bX_m)}{N\varepsilon^2}, 
    \label{eq:sample complexity}
\end{align}
where $\Var_{\cE}(\bX_m)$ is the variance of the random variable $\bX_m$ with the unitary ensemble $\cE$. 
To achieve the desired precision and confidence, it suffices to use $N\geq \Var_{\cE}(\bX_m) / (\delta\varepsilon^2)$ 
random unitaries drawn from the ensemble $\cE$.
Then, we focus on the variance $\Var_\cE(\bX_m)$. 
With the law of total variance, we have the following lemma (see proof in Appendix~\ref{app:general framework for DIPE}). 

\begin{lemma}
\label{lem:variance}
Given two quantum states $\rho, \sigma$ in $\cH_n$ and a unitary ensemble $\cE$, the variance of the random variable $\bX_m$ is 
\begin{align}
    \Var_\cE(\bX_m) 
    &= \sum_{i=1}^4 \Var_{\cE}^{(i)} (\rho,\sigma),
    \label{eq:variance}
\end{align}
where $\Var_{\cE}^{(1)} (\rho,\sigma)=-\tr^2[\rho\sigma]$,
\begin{equation}
\begin{split}
\Var_{\cE}^{(2)}(\rho,\sigma) 
&= \frac{1}{m^2} \tr\left[\cM_\cE^{(2)}(O^2) (\rho\ox\sigma)\right], \\
\Var_{\cE}^{(3)}(\rho,\sigma) 
&= \frac{m-1}{m^2} \bE_{U}  \bE f_\cE(\ba,\bb) \left[f_\cE(\ba',\bb) + f_\cE(\ba,\bb')\right], \\
\Var_{\cE}^{(4)}(\rho,\sigma) 
&= \frac{(m-1)^2}{m^2}\tr\left[\cM_\cE^{(4)}(O^{\ox 2}) (\rho\ox\sigma)^{\ox 2}\right]. \notag 
\end{split}
\end{equation}
\end{lemma}

As we can see, the variance depends on three main factors: the input states $\rho,\sigma$ and the unitary ensemble $\cE$. 
In particular, each term $\Var_{\cE}^{(k)}$ involves the $k$-moment channel of $\cE$ for $k\geq 2$, which is often difficult to compute analytically.

To date, rigorous state-dependent variance analysis has been established only for 
DIPE with a unitary $4$-design ensemble $\cF_n$. In this case, the \emph{worst-case variance} is given by
\begin{align}
\max_{\rho,\sigma} \Var_{\cF_n}(\bX_m) 
= \cO\left(\frac{2^n}{m^2} + \frac{1}{m} + \frac{1}{2^n}\right).
\end{align}
Hence, each platform requires $Nm = \Theta(\sqrt{2^n})$ state copies in the worst case.
See details in~\cite{anshu2022distributed} and Appendix~\ref{app:general framework for DIPE}. 
This exponential sample complexity highlights the intrinsic difficulty of DIPE, even with powerful unitary $4$-designs.
This naturally raises the question of \emph{whether the exponential hardness we established is overly pessimistic or rarely encountered in practical situations}.
In other words, \emph{how does DIPE perform for most states?}

\subsection{Average Sample Complexity}

To answer this question, it is necessary to analyze the average sample complexity, which captures the typical behavior for most states.
In various tasks, the average sample complexity is much lower than the worst-case complexity, suggesting that the task may not be as hard as the worst-case analysis indicates~\cite{jeon2025query}. 
Due to its importance, the average sample complexity has been widely studied in quantum learning theory, including state learning~\cite{bertoni2024shallow, ippoliti2023operator, akhtar2023scalable, hu2023classical, bu2024classical, hu2025demonstration} and channel learning~\cite{huang2021informationtheoretic, li2025nearly, jeon2025query}.

Specifically, in this work, we consider two common application scenarios of DIPE:
(i) estimating the purity of an unknown state, and
(ii) estimating the inner product between two unknown states.
A key observation from~\cite{anshu2022distributed} is that the sample complexity reaches its maximum when the unknown states are pure. 
Motivated by this, we define two types of average variances as follows. 

\begin{definition}[Average Variances]
\hfill\par 
\noindent\textbf{Case 1:} 
Let $\rho=\sigma = \ketbra{\psi}$, where $\ket{\psi}$ is a Haar random state. 
The average variance 1 is defined as 
\begin{align}
\Var^{a}_{\cE,1} &:= \bE_{\psi} \Var_{\cE}(\bX_m). 
\label{eq:average-case sample complexity 2}
\end{align}

\noindent\textbf{Case 2:} 
Let $\rho = \ketbra{\psi}$ and $\sigma = \ketbra{\phi}$, where $\ket{\psi}$ and $\ket{\phi}$ are two independent Haar random states. 
The average variance 2 is defined as 
\begin{align}
\Var^{a}_{\cE,2} &:= \bE_{\psi, \phi} \Var_{\cE}(\bX_m). 
\label{eq:average-case sample complexity}
\end{align}
\end{definition}

Based on these definitions, we provide the following theorem that relates the average variances of DIPE to the classical function; See proof and the concrete formulas in Appendix~\ref{app:average sample complexity}. 

\begin{theorem}
\label{the:average sample complexity}
Let $\cE$ be a Pauli-invariant ensemble with classical function $f_\cE$, the average variances defined in Eqs.~\eqref{eq:average-case sample complexity 2} and~\eqref{eq:average-case sample complexity} are given by
\begin{align}
\Var^{a}_{\cE,1} &= \cO\left(\frac{\norm{f_{\cE}}{2}^2}{2^n m^2} + \frac{\norm{f_{\cE}}{2}^2}{4^n m} + \frac{\norm{f_{\cE}}{2}^2}{8^n} + \frac{f^2_{\cE}(\bm{0}, \bm{0})}{4^n} - 1\right), \notag \\
\quad\Var^{a}_{\cE,2} &= \cO\left(\frac{\norm{f_{\cE}}{2}^2}{2^n m^2} + \frac{\norm{f_{\cE}}{2}^2}{4^n m} + \frac{\norm{f_{\cE}}{2}^2}{8^n}\right). \notag 
\end{align}
where $\norm{f_{\cE}}{2}^2 := \sum_{\ba}f_\cE^2(\ba,\bm{0})$. 
\end{theorem}
Theorem~\ref{the:average sample complexity} implies that once the classical function $f_{\cE}$ is known, the average variances can be computed directly.
Therefore, with the definition of $f_{\cT_n}$ and $f_{\cT_1^{\ox n}}$, 
we establish the following two lemmas. These lemmas characterize the average sample complexities for arbitrary global 
and local unitary $2$-design ensembles, respectively. Their proofs are detailed in Appendix~\ref{app:average sample complexity}.

\begin{lemma}
\label{lem:average variance of 2-design}
Let $\cT_n$ be a $2$-design ensemble, the average variances defined in Eqs.~\eqref{eq:average-case sample complexity 2} and~\eqref{eq:average-case sample complexity} are given by
\begin{align}
\Var^{a}_{\cT_n,1}, \Var^{a}_{\cT_n,2} &= \cO\left(\frac{2^n}{m^2} + \frac{1}{m} + \frac{1}{2^n}\right),
\end{align}
Consequently, the average sample complexity is $\cO(\sqrt{2^n})$.
\end{lemma}

\begin{lemma}
\label{lem:average variance of local 2-design}
Let $\cT_1^{\ox n}$ be a local unitary $2$-design ensemble, the average variances defined in Eqs.~\eqref{eq:average-case sample complexity 2} and~\eqref{eq:average-case sample complexity} are 
\begin{align}
\Var^{a}_{\cT_1^{\ox n},1},\Var^{a}_{\cT_1^{\ox n},2} &= \cO\left(\frac{2.5^n}{m^2} + \frac{1.25^n}{m} + 0.675^n\right), 
\end{align}
Consequently, the average sample complexity is $\cO(\sqrt{2.5^n})$.
\end{lemma}

Several important remarks are in order. 
First, the above two lemmas reveal that DIPE with both global and local unitary $2$-design ensembles is also exponentially hard for most states, highlighting the intrinsic difficulty of DIPE.
Second, our findings completely settle the average sample complexity of the global Clifford ensemble, as it is an instance of a global $2$-design ensemble.
Third, we have established an analytical upper bound for the average sample complexity of the local Clifford ensemble, which is a special case of a local $2$-design ensemble. 
Finally, we observe from both lemmas that the second term of variance, $\Var_{\cE}^{(2)}$,
is the primary factor driving scalability, an insight further corroborated by the numerical results presented in
Section~\ref{sec:numerical simulation} and Appendix~\ref{app:more experiments}.
This observation hints that the second moment can characterize the asymptotic state-dependent sample complexity.

In the next section, we consider an experimentally friendly unitary ensemble that interpolates between 
the local and global unitary $2$-design ensembles in terms of the average sample and the asymptotic state-dependent complexities.

\section{DIPE with Brickwork Ensembles}
\label{sec:cpe with brickwork circuits}

Here we consider DIPE with the brickwork ensemble~\cite{bertoni2024shallow, rozon2024optimala, ippoliti2023operator}, 
which is parameterized by one layer of local Clifford circuits and depth-$d$ two-local Clifford circuits, as shown in Fig.~\ref{fig: framework}. 
We denote the brickwork ensemble of depth $d$ as $\cB_d$. 
Notably, $\cB_0$ reduces to the local Clifford ensemble.
In the following, we first provide the classical function and average sample complexity of DIPE with brickwork ensembles. 
Then, we analyze the asymptotic state-dependent variance to understand the influence of depth.

\subsection{Classical Function and Average Variance}

First, we need to construct the corresponding classical function. 
The result is shown in the following lemma (see proof in Appendix~\ref{app:DIPE with random brickwork ensemble}).
\begin{lemma}
\label{lem:classical function of brikckwork}
Let $\cB_d$ be a brickwork ensemble, the classical function $f_d$ is given by 
\begin{align}
    f_{d}(\ba, \bb) = 2^n \prod_{s\in S} (-2)^{-2\delta_{\ba_s, \bb_s}}, 
\end{align}
where $S = \{(1, 2), \cdots, (n-1, n)\}$ if $d$ is odd, 
otherwise $S = \{(2, 3), \cdots, (n, 1)\}$, and $\ba_{(i,j)}$ is the $i$ and $j$-th bits of $\ba$. 
\end{lemma}

As shown in Lemma~\ref{lem:classical function of brikckwork}, the classical function is \emph{independent} of the depth and depends only on the parity of the depth. 
This is quite different from shallow shadows~\cite{bertoni2024shallow}, where the classical function varies with depth.
The reason is that applying the same random unitaries to both $\rho$ and $\sigma$ does not change the inner product $\tr[\rho\sigma]$, allowing us to ignore the influence of the former layers when constructing the classical function.
Therefore, it is reasonable to expect that the classical function of DIPE with brickwork ensembles is independent of depth;
see mathematical details in Appendix~\ref{app:DIPE with random brickwork ensemble}.

We now turn to analyzing the sample complexity and first consider the average variance.
Given the classical function defined in Lemma~\ref{lem:classical function of brikckwork}, we have the following lemma. 
\begin{lemma}
\label{lem:average variance of the brickwork}
Let $\cB_d$ be a brickwork ensemble, the average variances defined in Eqs.~\eqref{eq:average-case sample complexity 2} and~\eqref{eq:average-case sample complexity} are given by
\begin{align}
\Var^{a}_{\cB_d,1}, \Var^{a}_{\cB_d,2} &= \cO\left(\frac{2.18^n}{m^2} + \frac{1.09^n}{m} + 0.54^n\right). 
\end{align}
Consequently, the average sample complexity is $\cO(\sqrt{2.18^n})$.
\end{lemma}

Interestingly, Lemma~\ref{lem:average variance of the brickwork} implies that the average sample complexity of DIPE with
brickwork ensembles is independent of the depth $d$. 
This naturally raises the question: \emph{What role does the depth of the brickwork ensemble play in the performance of DIPE?}

\subsection{Asymptotic State-dependent Variance}

To investigate the influence of the depth, we then consider the asymptotic state-dependent variance, which is determined by the second term of the variance.
Define
\begin{align}
\Xi_{\rho,\sigma}(P) := \tr[P\rho]\tr[P\sigma], \quad 
\sum_{P\in\cP_n} \Xi_{\rho,\sigma}(P) = 2^n \tr[\rho\sigma]. \notag
\end{align}
We have the following lemma (see proof in Appendix~\ref{app:DIPE with random brickwork ensemble}).
\begin{lemma}
\label{lem:variance second term of shallow ensemble}
For the brickwork ensemble $\cB_d$, the corresponding classical function $f_d$, and states $\rho, \sigma$ in $\cH_n$, 
the second term of the variance is given by
\begin{align}
\Var_{\cB_d}^{(2)}(\rho, \sigma) 
= \frac{1}{2^n m^2}
\sum_{P\in\cP_n} \Xi_{\rho,\sigma}(P) \Upsilon_d(P), 
\label{eq:second term of variance brickwork}
\end{align}
where $\Upsilon_d(P):= \sum_{\ba\in\bZ_2^n} f_d^2(\ba, \bm{0}) h(\ba, P)$ and 
\begin{align}
h(\ba, P) := \bE_{U\sim\cB_d} \bra{\bm{0}}UPU^\dagger\ket{\bm{0}} \bra{\ba}UPU^\dagger\ket{\ba}.
\end{align}
\end{lemma}
As we can see, $\Upsilon_d(P)$ is hard to compute analytically. 
To address this, we first focus on $h(\ba, P)$ and rewrite it in the following form, 
\begin{align}
    h(\ba, P) 
    := \Pr\{UPU^\dagger\in\cZ^{\rm C}_{\ba}\} - \Pr\{UPU^\dagger\in\cZ^{\rm A}_{\ba}\}, 
    \label{eq:physical meaning of h(a,P)-main text}
\end{align} 
where $\cZ := \{I, Z\}^{\ox n}$, $X^{\ba} := \bigotimes X_i^{\ba_i}$, and
\begin{align}
\cZ^{\rm C}_{\ba} &:= \{P|P\in\pm\cZ, [P,X^{\ba}]=0\}, \\
\cZ^{\rm A}_{\ba} &:= \{P|P\in\pm\cZ, \{P,X^{\ba}\}=0\}. 
\end{align} 
Therefore, the physical meaning of $h(\ba, P)$ is the difference between the probabilities that $UPU^\dagger$ commutes or anti-commutes with $X^{\ba}$. 
Prior work~\cite[Lemma 5]{bertoni2024shallow} shows that $h(\bm{0}, P)$ admits a matrix product state (MPS) representation with a clear physical interpretation. 
Likewise, based on this physical meaning, we can represent $h(\ba, P)$ as a matrix product operator (MPO).
Furthermore, based on the special structure of the classical function $f_d$, we can also represent it as an MPS.
Therefore, we can combine these two tensor network representations to compute $\Upsilon_d(P)$, i.e., representing $\Upsilon_d(P)$ as an MPS, as shown in the following lemma. 

\begin{lemma}
\label{lem:MPO representation}
$\Upsilon_d(P)$ can be represented as a MPS with bond dimension at most $\cO(2^{d-1})$. 
For depth $d=\cO(\log n)$, it can be computed exactly in time $n^{\cO(1)}$.
\end{lemma}
    
The construction is detailed in Appendix~\ref{app:DIPE with random brickwork ensemble}. 
Numerical results in Appendix~\ref{app:more experiments} show that 
for Pauli operators $P\in\cP_n \setminus \{I^{\ox n}\}$, $\Upsilon_d(P)$ converges to $2^n$ as the depth $d$ increases.  
This behavior may be explained by statistical mechanical models~\cite{bertoni2024shallow, hunter-jones2019unitary} and operator spreading~\cite{ippoliti2023operator, hu2025demonstration}. 
This convergence phenomenon suggests that the asymptotic state-dependent sample complexity will converge to $\cO(\sqrt{2.18^n})$ for all state pairs as the depth increases, since $\Upsilon_d(P) \approx 2^n$ for all Pauli operators except the identity.
Furthermore, we observe that for some Pauli operators, $\Upsilon_d(P)$ decreases as the depth $d$ increases, while for others, $\Upsilon_d(P)$ increases with $d$.
This phenomenon indicates that not every state benefits from increasing depth, which is consistent with the fact that brickwork ensembles of different depths share the same average sample complexity.
Lastly, we numerically investigate the dependence of the fourth term of the variance on the number of qubits $n$, with detailed results also provided in Appendix~\ref{app:more experiments}.

\section{DIPE with Clifford Ensembles}
\label{sec:cpe with clifford}

We now consider DIPE with the global and local Clifford ensembles, which are two extreme cases of the brickwork ensemble $\cB_d$,
and compute their sample complexities. 
While the global and local Clifford ensembles are unitary $2$-design and local unitary $2$-design ensembles, respectively,
our analysis on the influence of circuit depth in Section~\ref{sec:cpe with brickwork circuits} motivates a more refined analysis.
We therefore focus on the state-dependent sample complexities of these two ensembles.

\subsection{Global Clifford Ensemble}

The global $n$-qubit Clifford ensemble $\Cl_n$ forms a unitary $3$-design~\cite{zhu2017multiqubit}. 
The classical function is $f_{\Cl_n} \equiv f_{\cT_n}$ and two average variances are given in Lemma~\ref{lem:average variance of 2-design}. 
Here, using the Schur-Weyl duality theory for the Clifford group~\cite{gross2021schur,chen2024nonstabilizerness}, 
we analyze the state-dependent variance and obtain the following.

\begin{theorem}
\label{the:global clifford variance}
For the global Clifford ensemble $\Cl_n$ and states $\rho,\sigma$ in $\cH_n$,
the variance of $\bX_m$ defined in Eq.~\eqref{eq:variance} satisfies 
\begin{align}
    \Var_{\Cl_n}(\bX_m) 
    = \cO\left(\frac{2^n}{m^2} + \frac{1}{m} + \frac{1}{2^n}\norm{\Xi_{\rho,\sigma}}{2}^2\right), \label{eq: global clifford bound}
\end{align}
where $\norm{\Xi_{\rho,\sigma}}{2}^2 := \sum_P \tr^2[P\rho]\tr^2[P\sigma]$.
Consequently, the worst-case sample complexity is $\Theta(\sqrt{2^n})$, where the matching lower bound has been proven in~\cite{anshu2022distributed}.
\end{theorem}
The proof is provided in Appendix~\ref{app:clifford}. 
Notably, this result shows that the global Clifford ensemble achieves \emph{a performance comparable} to that of the unitary $4$-design for all states $\rho$ and $\sigma$. 
Moreover, $\norm{\Xi_{\rho,\sigma}}{2}$ is a nonstabilizerness measure studied in~\cite{chen2024nonstabilizerness}. 
Theorem~\ref{the:global clifford variance} implies that nonstabilizerness can reduce the variance and improve the efficiency of DIPE with the global Clifford ensemble.
This phenomenon has also been observed in other tasks such as direct fidelity estimation~\cite{leone2023nonstabilizerness} and thrifty classical shadows~\cite{chen2024nonstabilizerness}.

\subsection{Local Clifford Ensemble}

Since the single-qubit Clifford ensemble $\Cl_1$ is a unitary $2$-design on $\cH_1$, 
the classical function of the local $n$-qubit Clifford ensemble $\Cl_1^{\ox n}$ satisfies $f_{\Cl_1^{\ox n}} \equiv f_{\cT_1^{\ox n}}$, 
and the average variances are given in Lemma~\ref{lem:average variance of local 2-design}.
For the state-dependent variance, we focus on the second and the fourth terms, as shown in the following theorem.

\begin{theorem}
\label{the:local clifford variance}
For the local Clifford ensemble $\Cl_1^{\ox n}$ and states $\rho,\sigma$ in $\cH_n$, 
the second term of the variance is given by
\begin{align}
    \Var_{\Cl_1^{\ox n}}^{(2)}(\rho,\sigma) 
    &= \frac{2.5^n}{m^2}\sum_{P\in\cP_n} \frac{\Xi_{\rho,\sigma}(P)}{5^{|P|}} 
    \leq \frac{3^n}{m^2}, 
    \label{eq: local clifford bound}
\end{align}
where $\Xi_{\rho,\sigma}(P) := \tr[P\rho]\tr[P\sigma]$ and $|P|$ is the Pauli weight. 
The upper bound in Eq.~\eqref{eq: local clifford bound} is achieved when $\rho=\sigma$ is a product state.
The fourth term of the variance is given by 
\begin{align}
    \Var_{\Cl_1^{\ox n}}^{(4)}(\rho,\sigma) 
    &= \frac{(m-1)^2}{4^n m^2} \sum_{P, Q\in \cP_n^P} \frac{\Xi_{\rho,\sigma}(P) \Xi_{\rho,\sigma}(Q)}{3^{-|\{i|P_i=Q_i\neq I\}|}}, \label{eq: local clifford 4-moment}
\end{align}
where $\cP_n^P := \left\{Q\in\cP_n | \forall i, |P_i| \cdot |Q_i| = 0 \text{ or } P_i = Q_i\right\}$. 

Consequently, if $\rho=\sigma$ is a stabilizer product state, the sample complexity is $\cO(\sqrt{4.5^n})$.
\end{theorem}

See Appendix~\ref{app:clifford} for the proof. 
From Theorem~\ref{the:local clifford variance}, we can find that the performance of DIPE with the local Clifford ensemble is also influenced by the nonstabilizerness. 
Unfortunately, it is challenging to derive a state-independent upper bound for Eq.~\eqref{eq: local clifford 4-moment}, and thus the worst-case sample complexity 
for local Clifford ensembles remains undetermined.

\begin{figure}[!htbp]
\centering
\begin{overpic}[width=1.0\linewidth]{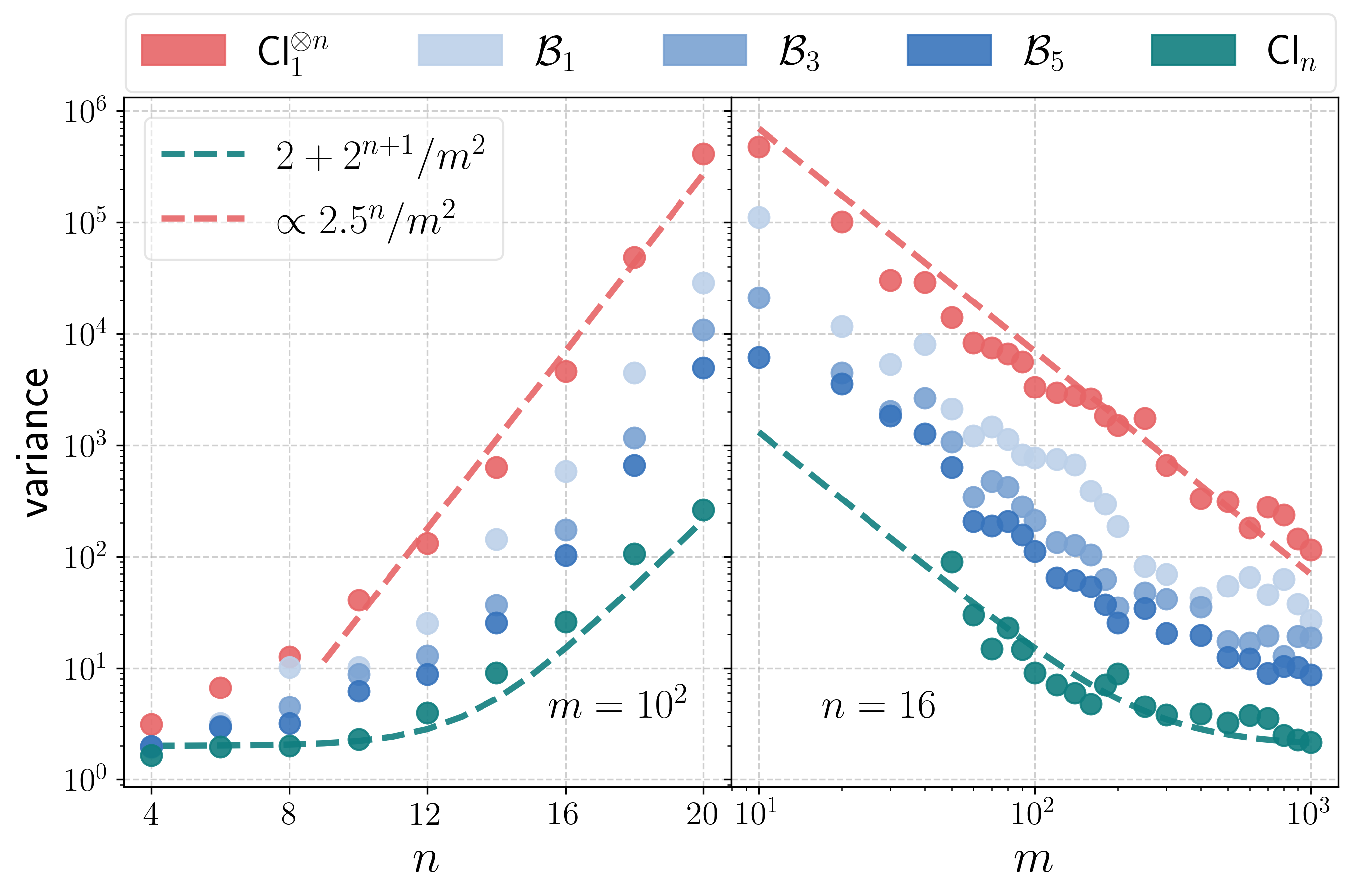}
\put(0, 60){(a)}
\end{overpic}
\begin{overpic}[width=1.0\linewidth]{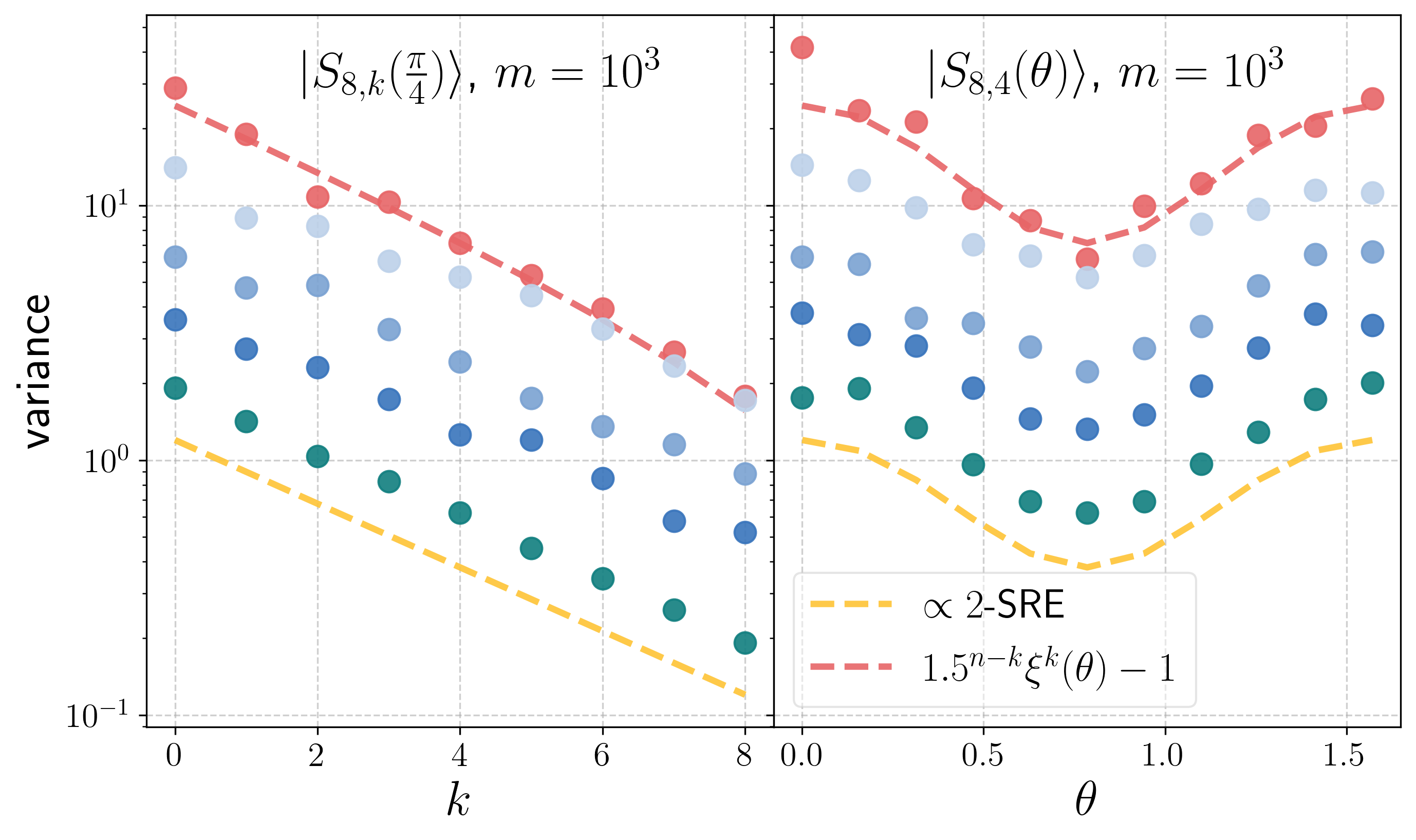}
\put(0, 55){(b)}
\end{overpic}
\caption{\raggedright
Numerical results of DIPE with different unitary ensembles: 
the local Clifford ensemble $\Cl_1^{\ox n}$, the brickwork ensemble $\cB_d$ ($d=1,3,5$), and the global Clifford ensemble $\Cl_n$.
The states are set as $\rho=\sigma$, 
\textbf{(a)} GHZ state, and
\textbf{(b)} $\ket{S_{8,k}(\theta)}$ defined in Eq.~\eqref{eq:definition of S state}.
Each data point is obtained with $10^2$ unitaries, $10^2$ state pairs, and $m$ shots. 
The green line is from Theorem~\ref{the:global clifford variance},  
the yellow line tracks the behavior of the stabilizer $2$-R\'enyi entropy ($2$-SRE), 
and the red line is from Theorem~\ref{the:local clifford variance}, where 
$\xi(\theta)$ is defined in Eq.~\eqref{eq:single-qubit 4-moment}. 
}
\label{fig: experiment}
\end{figure}

\section{DIPE with Approximate Unitary 4-design Ensembles}
\label{sec:cpe with approximate 4-design}

We now analyze the performance guarantee of DIPE using an $\varepsilon$-approximate $4$-design ensemble, denoted as $\tilde{\cF}_n$. 
Such ensemble can be constructed using $\cO(\log(n))$-depth circuits, and is defined as follows~\cite{schuster2025random}. 

\begin{definition}
An $n$-qubit unitary ensemble $\tilde{\cF}_n$ is an $\epsilon$-approximate unitary $4$-design if 
\begin{align}
(1-\epsilon)\cM^{(4)}_{\cF_n} \preceq \cM^{(4)}_{\tilde{\cF}_n} \preceq (1+\epsilon)\cM^{(4)}_{\cF_n},
\end{align}
where $\cF_n$ is an exact unitary $4$-design ensemble for $n$ qubits, $\cM_{\cE}$ is the $k$-moment channel defined in Eq.~\eqref{eq:k-moment channel}, and $\cA\preceq \cB$ denotes that $\cB - \cA$ is a completely-positive map. 
\end{definition}
As shown in~\cite{schuster2025random}, the approximation error $\epsilon$ can be exponentially suppressed by increasing the circuit depth.
In the following, we construct a biased estimator of $\tr[\rho\sigma]$, whose bias \emph{decreases exponentially} with circuit depth.

Recall that DIPE protocol consists of two steps: randomized measurements and classical post-processing.
Given unknown quantum states $\rho$ and $\sigma$, we sample unitaries from $\tilde{\cF}_n$, 
apply the corresponding randomized measurements with $m$ shots, 
and process the measurement outcomes using the classical function $f_{\cT_n}$,
resulting in the classical estimator $\tilde{\bX}_m$. 
By repeating this procedure for $N$ times, we obtain the biased estimator $\tilde{\omega}$ of $\tr[\rho\sigma]$. 
We then have the following theorem; See proof in Appendix~\ref{app:approximate design}.

\begin{theorem}
\label{the:DIPE with approximate design}
For DIPE with $\epsilon$-approximate unitary $4$-design ensemble $\tilde{\cF}_n$, the biased estimator $\tilde{\omega}$ with the classical function $f_{\cT_n}$ satisfies 
\begin{align}
|\tilde{\omega} - \tr[\rho\sigma]| \leq \epsilon \left(1 + \tr[\rho\sigma]\right) \leq 2\epsilon, 
\end{align}
and the variance of the classical estimator $\tilde{\bX}_m$ satisfies
\begin{align}
\Var[\tilde{\bX}_m] - \Var_{\cF_n}[\bX_m]
&\leq \cO\left(\frac{\epsilon\cdot 2^n}{m^2} + \epsilon \right).
\end{align}
\end{theorem}

As we can see, the bias decreases exponentially with the circuit depth, since $\epsilon$ can be exponentially suppressed~\cite{schuster2025random}.
Moreover, the variance of the classical estimator $\tilde{\bX}_m$ closely approximates that of the exact unitary $4$-design ensemble $\cF_n$ when $m = \cO(\sqrt{2^n})$.
This behavior is consistent with that of classical shadows using approximate unitary $3$-design ensembles~\cite{schuster2025random}, where the bias and variance exhibit similar characteristics.

\section{Numerical Simulation}
\label{sec:numerical simulation}
We now present numerical experiments for DIPE with different unitary ensembles, 
including $\Cl_1^{\ox n}$, $\cB_d$ ($d=1,3,5$), and $\Cl_n$. 
For each data point in Fig.~\ref{fig: experiment}, we sample $10^2$ unitary and $10^2$ pairs of states. 
We first set $\rho=\sigma=\ketbra{{\rm GHZ}_n}$, where $\ket{{\rm GHZ}_n}=(\ket{0}^{\ox n} + \ket{1}^{\ox n})/\sqrt{2}$. 
In Fig.~\ref{fig: experiment}(a), we vary $n$ from $4$ to $20$ with fixed $m=10^2$, and then vary $m$ from $10$ to $10^3$ with fixed $n=16$.
Two reference lines are included: the green line corresponds to $2+2^{n+1}/m^2$ from Theorem~\ref{the:global clifford variance}, 
and the red line represents a scaling of $\propto 2.5^n/m^2$ from Theorem~\ref{the:local clifford variance}.
These results demonstrate that increasing the depth of the brickwork ensemble dramatically suppresses the variance.

We then investigate the influence of nonstabilizerness by setting $\rho = \sigma = \ketbra{S_{8, k}(\theta)}$ and $m=10^3$, where
\begin{align}
    \ket{S_{n, k}(\theta)} 
    = \ket{0}^{\ox n-k} \ox \left[\frac{1}{\sqrt{2}}\left(\ket{0} + e^{i\theta}\ket{1}\right)\right]^{\ox k}, 
\end{align}
which has previously been studied in~\cite{bravyi2005universal, chen2024nonstabilizerness}. 
Defined that 
\begin{align}
    M_2 (n,k,\theta) 
    := 2^{n-k}\left(1 + \cos^4\theta + \sin^4\theta\right)^k, 
    \label{eq:definition of S state}
\end{align}
which serves as a widely used measure of stabilizerness, known as the stabilizer $2$-R\'enyi entropy ($2$-SRE), for the state $\ket{S_{n,k}(\theta)}$~\cite{leone2022stabilizer, chen2024nonstabilizerness}. 
In Fig.~\ref{fig: experiment}(b), we vary $k$ from $0$ to $8$ with fixed $\theta=\pi/4$, and vary $\theta$ from $0$ to $\pi/2$ with fixed $k=4$. 
For the local Clifford ensemble, we compute the variance using Eq.~\eqref{eq: local clifford 4-moment}, which is given by $1.5^{n-k} \xi^{k}(\theta)-1$,
where $\tr^2[\rho\sigma]=1$,  $\tr[\cM_{\Cl_1}^{(4)}(O^{\ox 2})\ketbra{0}^{\ox 4}] = 1.5$, and 
\begin{align}
    \xi(\theta) := \tr[\cM_{\Cl_1}^{(4)}(O^{\ox 2})\ketbra{S_{1,1}(\theta)}^{\ox 4}]. \label{eq:single-qubit 4-moment}
\end{align}
For reference, we also plot the values of $M_2(8, k, \theta)$ in each subfigure. 
Our results reveal that, across all ensembles considered, the variance shows a strong positive correlation with $M_2(8, k, \theta)$, validating the trend established in Theorem~\ref{the:global clifford variance} for the global Clifford ensemble $\Cl_n$. 
More numerical simulation results can be found in Appendix~\ref{app:more experiments}.

\begin{table*}[!t]
\centering
\renewcommand{\arraystretch}{1.5} 
\setlength{\tabcolsep}{0pt} 
\setlength\heavyrulewidth{0.3ex}  
\begin{tabular}{C{0.25\linewidth} C{0.2\linewidth} C{0.2\linewidth}}
\toprule
& \multicolumn{2}{c}{\textbf{Average Sample Complexity}} \\ \cmidrule(l){2-3} 
\multirow{-2}{*}{\textbf{Unitary Ensemble}} & \textbf{$\rho=\sigma=\ketbra{\psi}$} & \textbf{$\rho=\ketbra{\psi}, \sigma = \ketbra{\phi}$} \\ \midrule
\textbf{unitary $2$-design $\cT_n$} & $\cO(\sqrt{2^n})$ & $\cO(\sqrt{2^n})$ \\
\textbf{local unitary $2$-design $\cT_1^{\ox n}$} & $\cO(\sqrt{2.5^n})$ & $\cO(\sqrt{2.5^n})$  \\
\textbf{brickwork $\cB_d$} & $\cO(\sqrt{2.18^n})$ & $\cO(\sqrt{2.18^n})$\\ 
\bottomrule
\end{tabular}
\caption{
The average sample complexity of DIPE with different unitary ensembles as a function of the number of qubits $n$,
where $\ket{\psi}$ and $\ket{\phi}$ are independent Haar random states.}%
\label{tab:comparison}
\end{table*}

\begin{table*}[!t]
\centering
\renewcommand{\arraystretch}{1.5} 
\setlength{\tabcolsep}{0pt} 
\setlength\heavyrulewidth{0.3ex}  
\begin{tabular}{C{0.2\linewidth} C{0.3\linewidth} C{0.25\linewidth} C{0.15\linewidth}}
\toprule
\textbf{Unitary Ensemble} & \textbf{State-dependent Variance} & \textbf{Worst-case Sample Complexity} & 
 \\ \midrule
    \textbf{unitary $4$-design $\cF_n$} 
&   $\displaystyle\cO\left(\frac{2^n}{m^2} + \frac{1}{m} + \frac{(1+\tr^2[\rho\sigma])^2}{2^n}\right)$ 
&   $\displaystyle\Theta(\sqrt{2^n})$ 
&   Anshu et al., STOC (2022) \\
    \textbf{global Clifford $\Cl_n$} 
&   $\displaystyle\cO\left(\frac{2^n}{m^2} + \frac{1}{m} + \frac{\norm{\Xi_{\rho,\sigma}}{2}^2}{2^n}\right)$ 
&   $\displaystyle\Theta(\sqrt{2^n})$  
&   This work \\
\bottomrule
\end{tabular}
\caption{%
The state-dependent variance of DIPE with different ensembles as a function of the number of qubits $n$, the number of shots $m$, and the quantum states $\rho$, $\sigma$.
Each worst-case sample complexity is obtained by maximizing the corresponding variance over all possible state pairs $(\rho,\sigma)$.%
}%
\label{tab:comparison variance}
\end{table*}

\section{Conclusions}
\label{sec:conclusions}
We presented the general requirements for DIPE, enabling the use of broader types of unitary ensembles to realize the protocol. 
Focusing on the average sample complexity, we showed that DIPE with the unitary $2$-design ensemble requires $\Theta(\sqrt{2^n})$ state copies on average, which is optimal.
We then extended our analysis to ensembles below unitary $2$-designs, as summarized in Table~\ref{tab:comparison}.
Specifically, we proved that DIPE with the local unitary 2-design requires $\mathcal{O}(\sqrt{2.5^n})$ copies on average, while the brickwork ensemble $\cB_d$ achieves $\mathcal{O}(\sqrt{2.18^n})$, which is independent of circuit depth.
To investigate the influence of depth, we developed a tensor network approach to compute the asymptotic state-dependent variance. 
We further analyzed the state-dependent sample complexity for the global and local Clifford ensembles.
For the global Clifford ensemble, DIPE requires $\Theta(\sqrt{2^n})$ copies for all $n$-qubit states $\rho$ and $\sigma$, achieving performance comparable to that of a unitary $4$-design.
In contrast, DIPE with the local Clifford ensemble requires $\mathcal{O}(\sqrt{4.5^n})$ copies for stabilizer product states.
We also showed that the nonstabilizerness of states enhances the performance of DIPE with the global and local Clifford ensembles. 
For DIPE with an $\varepsilon$-approximate unitary $4$-design ensemble, we constructed a biased estimator whose bias decreases exponentially with circuit depth and whose variance closely approximates that of an exact unitary $4$-design ensemble when using $\mathcal{O}(\sqrt{2^n})$ shots.
A summary of the proven state-dependent variances is in Table~\ref{tab:comparison variance}.

Many questions remain open.
For example, current DIPE and cross-platform verification protocols assume that both platforms implement the same unitary, 
an assumption that may not hold in practice due to hardware imperfections. 
This motivates the development of more robust protocols. 
As mentioned before, the worst-case sample complexity for DIPE with the local Clifford ensemble remains unknown.
One possible approach is to use tools from tomography with local Clifford ensembles~\cite{acharya2025pauli, grewal2026pauli}.
Additionally, it would also be interesting to explore the worst-case sample complexity of state learning tasks, including classical shadows and DIPE, with various unitary ensembles.

\emph{Note added.}
This work was submitted to \href{http://aqis-conf.org/2025/}{AQIS 2025} on April 25, 2025 and was selected for an oral presentation.
We became aware of a related work by Wu \emph{et al.}~\cite{wu2025state},
submitted to arXiv on June 2, 2025, during the final preparation of this manuscript for arXiv.
While their study explores DIPE with the local unitary $4$-design ensemble, 
our work focuses on a broader and distinct range of unitary ensembles.

\section*{Acknowledgments}
This work was supported by 
the National Natural Science Foundation of China (Grant Nos. 62471126),  the Jiangsu Key R\&D Program Project (Grant No. BE2023011-2), 
the SEU Innovation Capability Enhancement Plan for Doctoral Students (Grant No. CXJH\_SEU 24083), 
the National Key R\&D Program of China (Grant No. 2022YFF0712800), 
the Fundamental Research Funds for the Central Universities (Grant No. 2242022k60001), 
and the Big Data Computing Center of Southeast University.


\input{output.bbl}

\makeatletter
\newcommand{\appendixtitle}[1]{\gdef\@title{#1}}
\newcommand{\appendixauthor}[1]{\gdef\@author{#1}}
\newcommand{\appendixaffiliation}[1]{\gdef\@affiliation{#1}}
\newcommand{\appendixdate}[1]{\gdef\@date{#1}}
\makeatother

\makeatletter%
\newcommand{\appendixmaketitle}{%
\begin{center}%
\vspace{0.4in}%
{\Large \@title \par}%
\end{center}%
\par%
}%
\makeatother%

\setcounter{secnumdepth}{2}
\appendix
\widetext
\newpage

\appendixtitle{\bf 
Supplemental Material for\\``\thetitle''}
\appendixmaketitle
\vspace{0.2in}

In this Supplementary Material, we elaborate on details omitted from the main text, specifically:
\begin{itemize}
\item \textbf{Appendix~\ref{app:general framework for DIPE}}: We present the \emph{general framework for distributed inner product estimation (DIPE)}, including the full protocol and a detailed analysis of the state-dependent sample complexity.
\item \textbf{Appendix~\ref{app:average sample complexity}}: We analyze the average sample complexity by first relating it to a classical function, and then deriving analytic results for arbitrary global and local unitary $2$-design ensembles.
\item \textbf{Appendix~\ref{app:DIPE with random brickwork ensemble}}: We present the details on \emph{DIPE with the brickwork ensemble}.
\item \textbf{Appendix~\ref{app:clifford}}: We focus on \emph{DIPE with Clifford ensembles}, providing analysis of the state-dependent variances and sample complexities for both
local and global Clifford ensembles. 
\item \textbf{Appendix~\ref{app:approximate design}}: We discuss the performance guarantee of DIPE with \emph{$\varepsilon$-approximate $4$-design}, which can be constructed with $\cO(\log(n))$-depth circuits~\cite{schuster2025random}. 
\item \textbf{Appendix~\ref{app:more experiments}}: We show more numerical results. 
\item \textbf{Appendix~\ref{app:useful lemmas}}: We gather useful lemmas for our proof, concerning the properties of unitary designs, Clifford ensembles, and Haar random states. Some of these are from literature, while others are new and may be of independent interest.
\end{itemize}

\section{General framework for DIPE}
\label{app:general framework for DIPE}

\subsection{Classical Estimator}

Before introducing the general framework for DIPE, we first define a classical estimator. 
It is worth noting that our definition here generalizes the classical collision estimator presented in Ref.~\cite{anshu2022distributed}. 
We give a more general form of the classical estimator, which proves particularly useful for analyzing the sample complexity across various unitary ensembles.

\begin{definition}
\label{def:classical estimator}
Given samples $\bm{a}_1,\cdots, \bm{a}_m\sim p$ and $\bm{b}_1, \cdots, \bm{b}_m\sim q$ from two discrete distributions $p$ and $q$, respectively, 
a classical estimator is defined as 
\begin{align}
    \bX_m := \frac{1}{m^2}\sum_{i,j=1}^m f(\ba_i, \bb_j), 
    \label{eq:classical estimator}
\end{align}
where $f$ is a classical function.
\end{definition}

Now we analyze the expectation and variance of this estimator, 
which are crucial for understanding the requirements on the classical function $f$ as well as the sample complexity of DIPE.
The following lemma generalizes~\citep[Lemmas 15 and 16]{anshu2022distributed}, which focuses on the classical collision estimator, 
to more general classical estimators.

\begin{lemma}
\label{lem:expectation and variance of classical estimator}
The expectation of $\bX_m$ is given by 
\begin{align}
    \bE_{\ba,\bb} \bX_m = \sum_{\ba,\bb} p(\ba)q(\bb) f(\ba,\bb). 
\end{align}
The variance of $\bX_m$ is given by
\begin{align}
    \Var_{\ba,\bb}(\bX_m) 
    =& \frac{1}{m^2}\bE f^2(\ba,\bb)
    + \frac{m-1}{m^2}\left[\bE_{\ba,\bb} f(\ba,\bb) \left(\bE_{\ba'} f(\ba',\bb) + \bE_{\bb'} f(\ba,\bb')\right) \right]
    + \left[\frac{(m-1)^2}{m^2}-1\right]\left(\bE f(\ba,\bb)\right)^2. 
    \label{eq:variance of classical estimator}
\end{align}
\end{lemma}
\begin{proof}
First, we compute the expectation:
\begin{align}
    \bE_{\ba,\bb} \bX_m 
    = \frac{1}{m^2} \sum_{i,j=1}^m \bE f(\ba_i,\bb_j) 
    = \sum_{\ba,\bb} p(\ba)q(\bb) f(\ba,\bb).
\end{align}
Next, we compute the variance. 
In general, for the sum of $N$ random variables $\{\bY_i\}_{i=1}^N$, the variance is given by
\begin{align}
    \Var\left(\sum_{i=1}^N \bY_i\right) 
    = \sum_{i=1}^{N}\Var(\bY_i) + \sum_{i\neq j}{\rm Cov}(\bY_i, \bY_j), 
\end{align}
where ${\rm Cov}(\cdot,\cdot)$ is the covariance. 
Therefore, we have 
\begin{align}
    \Var_{\ba,\bb}(\bX_m) 
    &= \frac{1}{m^4} \sum_{i,j}\Var[f(\ba_i,\bb_j)] 
    + \frac{1}{m^4} \sum_{i\neq k, j=l} {\rm Cov}[f(\ba_i,\bb_j), f(\ba_k, \bb_j)] \notag \\
    &+ \frac{1}{m^4} \sum_{i=k, j\neq l} {\rm Cov}[f(\ba_i,\bb_j), f(\ba_i, \bb_l)] + 
    \frac{1}{m^4} \sum_{i\neq k, j\neq l} {\rm Cov}[f(\ba_i,\bb_j), f(\ba_k, \bb_l)] \\
    &= \frac{1}{m^2}\bE f^2(\ba,\bb)
    + \frac{m-1}{m^2}\left[\bE_{\ba,\bb} f(\ba,\bb) \left(\bE_{\ba'} f(\ba',\bb) + \bE_{\bb'} f(\ba,\bb')\right) \right]
    + \left[\left(\frac{m-1}{m}\right)^2-1\right]\left(\bE f(\ba,\bb)\right)^2. 
\end{align}
\end{proof}

\subsection{Requirement for the Classical Function: Proof and Examples}

We now present the detailed procedure of DIPE~\cite{anshu2022distributed}, which is summarized as follows. 

\begin{algorithm}[H]
\caption{Distributed Inner Product Estimation (DIPE)}
\begin{algorithmic}[1] \label{alg:general framework}
\REQUIRE A unitary ensemble $\cE=(\cU, \mu)$ \\
\hskip1.5em number of sampled unitaries $N$ \\
\hskip1.5em measurement shots $m$ \\
\hskip1.5em $Nm$ copies of unknown states $\rho$ and $\sigma$ on two platforms 
\ENSURE an estimator of $\omega := \tr[\rho\sigma]$
\FOR{$t = 1, \cdots, N$}
    \STATE sample a unitary $U\sim\cE$ according to $\mu$
    \STATE measure $m$ copies of $\rho$ in the basis $\{U^\dagger\ketbra{\bm{x}}U\}_{\bx}$ and obtain $A = \{\bm{a}_1, \cdots, \bm{a}_m\}$, 
                where $\{\ketbra{\bx}\}_{\bx}$ is the computation basis
    \STATE measure $m$ copies of $\sigma$ in the basis $\{U^\dagger\ketbra{\bm{x}}U\}_{\bx}$ and obtain $B = \{\bm{b}_1, \cdots, \bm{b}_m\}$ 
    \STATE compute the classical estimator defined in Eq.~\eqref{eq:classical estimator} using $A$ and $B$, denoted by $\bX_m^{(t)}$
\ENDFOR
\RETURN $\hat{\omega} = \sum_t \bX_m^{(t)} / N$
\end{algorithmic}
\end{algorithm}

To ensure that $\hat{\omega}$ is an \emph{unbiased estimator} of $\tr[\rho\sigma]$, 
the classical function $f$ must be carefully chosen.
Obviously, the classical function is highly related with the random unitary ensemble $\cE$. 
We prove Lemma~\ref{lem:requirement for classical function} in the main text as follows. 
\begin{proof}[Proof of Lemma~\ref{lem:requirement for classical function}]
We first compute the expectation of random variable: 
\begin{align}
    \bE\; \bX_m 
    &= \bE_{U\sim\cE} \frac{1}{m^2}\sum_{i,j=1}^m f(\bm{a}_i, \bm{b}_j) \\
    &= \bE_{U\sim\cE} \sum_{\ba,\bb} p_U(\ba)q_U(\bb) f(\ba,\bb) \\
    &= \tr\left[\cM_{\cE}^{(2)}(O)(\rho\ox\sigma)\right], 
\end{align}
where 
\begin{align}
    p_U(\ba) := \bra{\ba}U\rho U^\dagger\ket{\ba}, \quad 
    q_U(\bb) := \bra{\bb}U\sigma U^\dagger\ket{\bb}, \quad 
    O := \sum_{\ba,\bb} f(\ba,\bb) \ketbra{\ba\bb}. 
\end{align}
We find that $\bE_{U\sim\cE} \bX_m = \tr[\rho\sigma]$ hold for all $\rho,\sigma$ in $\cH_n$ if and only if
\begin{align}
    \cM_{\cE}^{(2)}(O) 
    = \frac{1}{2^n} \sum_{P\in\cP_n} P\ox P
    = \bigotimes_{i=1}^n \bS_i, 
\end{align}
where $\bS_i$ is the SWAP operator on $i$-th qubit of two states. 
In other words, we require the classical function satisfy 
\begin{align}
    \tr[\cM_\cE^{(2)}(O)(P\ox P')] = 
    \begin{cases}
        2^n, & P = P' \\
        0, &\text{otherwise}
    \end{cases},
\end{align}
for all $P,P'\in\cP_n$. 
\end{proof}

We then consider the Pauli-invariant ensemble. 

\begin{proof}[Proof of Lemma~\ref{lem:classical function requirement of Pauli-invariant ensemble}]
Recall Lemma~\ref{lem:requirement for classical function}, for $P,P'\in\cP_n$, if $\cE$ is a Pauli-invariant ensemble, we have 
\begin{align}
\tr[\cM_\cE(O)(P\ox P')] 
&= \bE_{U\sim\cE} \tr\left[\left(\sum_{\ba,\bb\in \bZ_2^{n}} f(\ba,\bb) U^{\dagger \ox 2} \ketbra{\ba\bb} U^{\ox 2}\right)(P\ox P')
\right] \\
&= \sum_{\ba,\bb\in \bZ_2^{n}} f(\ba,\bb) \bE_U \bra{\ba}U PU^\dagger\ket{\ba}
\bra{\bb}U P'U^\dagger\ket{\bb} \\
&= 2^n\sum_{\ba\in \bZ_2^{n}} f(\ba, \bm{0}) \bE_U \bra{\bm{0}}U PU^\dagger\ket{\bm{0}}\bra{\ba}U P'U^\dagger\ket{\ba}, 
\end{align}
where the last equality follows from the Pauli-invariance of the ensemble $\cE$, 
namely that $\mu(Z^{\ba} U) = \mu(U)$ for all $\ba \in \bZ_2^n$. 
This implies that the classical function $f(\ba, \bb)$ depends only on the bitwise XOR of $\ba$ and $\bb$:
\begin{align}
    f(\ba, \bb) = f(\ba\oplus\bb, \bm{0}). 
\end{align} 
\end{proof}
Based on this property of the Pauli-invariant ensemble $\cE$, we can also obtain the following results: 
\begin{align}
f_{\cE}(\ba, \bb) &= f_\cE(\bb, \ba), \quad \forall\; \ba,\bb, \\
\tr[\cM_{\cE}^{(2)}(O)] = 2^n \quad &\Rightarrow\quad 
\sum_{\ba} f_{\cE}(\ba, \bb) = 1,  \quad \forall\; \bb \label{eq:sum property of pauli-invariant classical function}
\end{align}

\subsubsection{Special Cases}

Then, we show two typical examples of the classical function, 
which are seminally constructed in~\cite{elben2020crossplatform}.

\begin{example}
If $\cE=\cT_n$ is a unitary $2$-design ensemble, we have 
\begin{align}
    f_{\cT_n}(\ba, \bb) 
    = \begin{cases}
        2^n, & \ba = \bb, \\
        -1, & \ba \neq \bb. 
    \end{cases}
\end{align}
\end{example}

\begin{proof}
With the definition of $f_{\cT_n}$, we have 
\begin{align}
    O &= \sum_{\ba,\bb} f_{\cT_n}(\ba,\bb) \ketbra{\ba\bb} 
    = (2^n + 1) \sum_{\ba} \ketbra{\ba\ba} - \1. 
\end{align}
Then, we have 
\begin{align}
    \cM_{\cT_n}^{(2)}(O) 
    = \frac{\tr[O] - \tr\left[(\bigotimes\bS) O\right]/2^n}{4^n-1} \1 + 
    \frac{\tr[(\bigotimes\bS) O] - \tr[O]/2^n}{4^n-1} \bigotimes_{i=1}^n \bS_i
    = \bigotimes_{i=1}^n \bS_i, 
\end{align}
where we use Lemma~\ref{lem:unitary 2-design expand}.
\end{proof}

\begin{example}
If $\cE=\cT_1^{\ox n}$ is a local unitary $2$-design, we have 
\begin{align}
    f_{\cT_1^{\ox n}}(\bm{a}, \bm{b}) = 2^n\cdot(-2)^{-\cD(\bm{a}, \bm{b})},
\end{align}
where $\cD(\bm{a}, \bm{b})$ is the Hamming distance between $\bm{a}$ and $\bm{b}$. 
\end{example}

\begin{proof}
With the definition of $\cT_1^{\ox n}$, we have 
\begin{align}
O &= \sum_{\ba,\bb}f_{\cT_1^{\ox n}}(\ba,\bb) \ketbra{\ba\bb} 
= 2^n\sum_{\ba,\bb} \bigotimes_{i=1}^{n} (-2)^{-\cD(\ba_i,\bb_i)}\ketbra{\ba_i\bb_i} \\
&= \bigotimes_{i=1}^n \left(2\ketbra{00} + 2\ketbra{11} - \ketbra{01} - \ketbra{10}\right) 
\end{align}
Then, with Lemma~\ref{lem:unitary 2-design expand}, we have 
\begin{align}
\cM_{\cT_1^{\ox n}}^{(2)}(O) 
&= \bigotimes_{i=1}^n \cM_{\cT_1}^{(2)}\left(2\ketbra{00} + 2\ketbra{11} - \ketbra{01} - \ketbra{10}\right)
= \bigotimes_{i=1}^n \bS_i.
\end{align}
\end{proof}

\subsection{Sample Complexity}

Now we consider the sample complexity of DIPE. 
By Chebyshev's inequality, we can only focus on the variance of random variable $\bX_m$. 
\begin{proof}[Proof of Lemma~\ref{lem:variance}]
With the law of total variance, we have 
\begin{align}
\Var_\cE(\bX_m)
&= \bE_{U\sim\cE} \Var_{\ba,\bb}[\bX_m|U] + 
\Var_{U\sim\cE}\left[\bE_{\ba,\bb}(\bX_m|U)\right].
\end{align}
With Eq.~\eqref{eq:variance of classical estimator} and
\begin{align}
\Var_{U\sim\cE} [\bE_{\ba,\bb}(\bX_m|U)] 
= \bE_{U\sim\cE} \left[\bE_{\ba,\bb} \bX_m\right]^2 - \tr^2[\rho\sigma], 
\end{align}
we have 
\begin{align}
\Var_\cE(\bX_m)
&= \frac{1}{m^2} \bE_{U\sim\cE, \ba,\bb} f^2(\ba,\bb) + \left(\frac{m-1}{m}\right)^2\bE_{U\sim\cE}(\bE_{\ba,\bb} \bX_m)^2 - \tr^2[\rho\sigma] \notag \\
&\quad + \frac{m-1}{m^2} \bE_{U\sim\cE} \left[\bE_{\ba,\bb} f(\ba,\bb) \left(\bE_{\ba'} f(\ba',\bb) + \bE_{\bb'} f(\ba,\bb')\right)\right] \\
&= \frac{1}{m^2} \tr\left[\cM_\cE^{(2)}(O^2) (\rho\ox\sigma)\right] + \left(\frac{m-1}{m}\right)^2\tr\left[\cM_\cE^{(4)}(O^{\ox 2}) (\rho\ox\sigma)^{\ox 2}\right] - \tr^2[\rho\sigma] \notag \\
&\quad + \frac{m-1}{m^2} \bE_{U\sim\cE} \left[\bE_{\ba,\bb} f(\ba,\bb) \left(\bE_{\ba'} f(\ba',\bb) + \bE_{\bb'} f(\ba,\bb')\right)\right] \\
&=: \sum_{i=1}^4 \Var_{\cE}^{(i)} (\rho,\sigma), \label{eq:variance appendix}
\end{align}
where 
\begin{align}
\Var_{\cE}^{(1)}(\rho,\sigma) 
&= - \tr^2[\rho\sigma], \\
\Var_{\cE}^{(2)}(\rho,\sigma) 
&= \frac{1}{m^2} \tr\left[\cM_\cE^{(2)}(O^2) (\rho\ox\sigma)\right], \label{eq:variance 2}\\
\Var_{\cE}^{(3)}(\rho,\sigma) 
&= \frac{m-1}{m^2} \bE_{U\sim\cE} \left[\bE_{\ba,\bb} f_\cE(\ba,\bb) \left(\bE_{\ba'} f_\cE(\ba',\bb) + \bE_{\bb'} f_\cE(\ba,\bb')\right)\right], \\
\Var_{\cE}^{(4)}(\rho,\sigma) 
&= \left(\frac{m-1}{m}\right)^2\tr\left[\cM_\cE^{(4)}(O^{\ox 2}) (\rho\ox\sigma)^{\ox 2}\right].  
\end{align}
\end{proof}

We find that the value of $\Var_{\cE}^{(k)}(\rho,\sigma)$ depends on the $k$-moment of the unitary ensemble $\cE$ for $k\geq 2$, leading to the following lemmas. 
Notably, here we consider the general classical estimator and obtain the same results shown in~\cite{anshu2022distributed}, where the authors consider the classical collision estimator. For completeness and later references, we also provide proofs here.

\begin{lemma}
\label{lem:variance 2 of unitary 2-design}
If the random unitary ensemble $\cE$ is a unitary $2$-design, 
\begin{align}
    \Var_\cE^{(2)}(\rho, \sigma) 
    = \frac{2^n + (2^n-1)\tr[\rho\sigma]}{m^2} 
    = \cO\left(\frac{2^n}{m^2}\right). 
\end{align}
\end{lemma}

\begin{proof}[Proof of Lemma~\ref{lem:variance 2 of unitary 2-design}]
If the random unitary ensemble $\cE$ is a unitary $2$-design, we have 
$f = f_{\cT_n}$ and 
\begin{align}
    O^2 
    = \sum_{\ba,\bb} f^2(\ba,\bb)\ketbra{\ba\bb} 
    = (4^n-1)\sum_{\ba} \ketbra{\ba\ba} + \1. 
\end{align}
Therefore, we have 
\begin{align}
    \Var_\cE^{(2)}(\rho, \sigma) 
    &= \frac{1}{m^2} \tr\left[\cM_\cE^{2}(O^2) (\rho\ox\sigma)\right] \\
    &= \frac{1}{m^2} \bE_{U\sim\cE} \tr\left[(4^n-1)\sum_{\ba} 
    U^{\dagger\ox 2} \ketbra{\ba\ba}U^{\ox 2} (\rho\ox\sigma)\right] 
    + \frac{1}{m^2}\\
    &= \frac{1}{m^2} \frac{4^n-1}{2^n(2^n+1)}\cdot 2^n (1+\tr[\rho\sigma]) + \frac{1}{m^2} \\
    &= \frac{2^n + (2^n-1)\tr[\rho\sigma]}{m^2} 
    = \cO\left(\frac{2^n}{m^2}\right). \tag*{Lemma~\ref{lem:properties of unitary design}}
\end{align}
\end{proof}

\begin{lemma}[Lemma 16 of~\cite{anshu2022distributed}]
\label{lem:variance 3 of unitary 3-design}
If the random unitary ensemble $\cE$ is a unitary $3$-design, 
\begin{align}
    \Var_\cE^{(3)}(\rho, \sigma) = \cO\left(\frac{1}{m}\right). 
\end{align}
\end{lemma}

\begin{proof}[Proof of Lemma~\ref{lem:variance 3 of unitary 3-design}]
If the random unitary ensemble $\cE$ is a unitary $3$-design, we have 
$f = f_{\cT_n}$ and 
\begin{align}
    \Var_\cE^{(3)}(\rho, \sigma)
    &= \frac{m-1}{m^2} \bE_{U\sim\cE} \left[\bE_{\ba,\bb,\ba'} f(\ba,\bb) f(\ba',\bb) + \bE_{\ba,\bb,\bb'} f(\ba,\bb) f(\ba,\bb') \right]. 
\end{align}
We consider one term first, 
\begin{align}
    \bE_{U\sim\cE} \bE_{\ba,\bb,\ba'} f(\ba,\bb) f(\ba',\bb) 
    &= \bE_{U\sim\cE} \sum_{\ba,\bb} \bra{\ba}U\rho U^\dagger\ket{\ba} 
    \bra{\bb}U\sigma U^\dagger\ket{\bb} f(\ba,\bb) \sum_{\ba'} \bra{\ba'}U\rho U^\dagger\ket{\ba'} f(\ba',\bb) \\
    &= \bE_{U\sim\cE} \sum_{\ba,\bb} \bra{\ba}U\rho U^\dagger\ket{\ba} 
    \bra{\bb}U\sigma U^\dagger\ket{\bb} f(\ba,\bb) 
    \left[(2^n+1) \bra{\bb}U\rho U^\dagger\ket{\bb} - 1\right] \\
    &= (2^n+1) \bE_{U\sim\cE} \sum_{\ba,\bb} \bra{\ba}U\rho U^\dagger\ket{\ba} \bra{\bb}U\sigma U^\dagger\ket{\bb} \bra{\bb}U\rho U^\dagger\ket{\bb} f(\ba,\bb)  - \tr[\rho\sigma] \\
    &= (2^n+1)^2\sum_{\ba}\bE_{U\sim\cE} \bra{\ba}U\rho U^\dagger\ket{\ba}^2 
    \bra{\ba}U\sigma U^\dagger\ket{\ba} \notag \\
    &- \sum_{\ba,\bb} \bE_{U\sim\cE} \bra{\ba}U\rho U^\dagger\ket{\ba} 
    \bra{\bb}U\sigma U^\dagger\ket{\bb} \bra{\bb}U\rho U^\dagger\ket{\bb} - \tr[\rho\sigma] \\
    &= (2^n+1)^2\sum_{\ba}\bE \bra{\ba}U\rho U^\dagger\ket{\ba}^2 
    \bra{\ba}U\sigma U^\dagger\ket{\ba} - \sum_{\bb} \bE \bra{\bb}U\sigma U^\dagger\ket{\bb} \bra{\bb}U\rho U^\dagger\ket{\bb} - \tr[\rho\sigma] \\
    &= \cO(1), 
\end{align}
where the last line use Lemma~\ref{lem:properties of unitary design}. 
Likewise, 
\begin{align}
    \bE_{U\sim\cE} \bE_{\ba,\bb,\bb'} f(\ba,\bb) f(\ba,\bb')
    = \cO(1). 
\end{align}
\end{proof}

\begin{lemma}[Lemma 16 of~\cite{anshu2022distributed}]
\label{lem:variance 4 of unitary 4-design}
If the random unitary ensemble $\cE$ is a unitary $4$-design, 
\begin{align}
    \Var_\cE^{(1)}(\rho, \sigma) + \Var_\cE^{(4)}(\rho, \sigma) 
    = \cO\left(\frac{1}{2^n}\right). 
\end{align}
\end{lemma}

\begin{proof}[Proof of Lemma~\ref{lem:variance 4 of unitary 4-design}]
If the random unitary ensemble $\cE$ is a unitary $4$-design, we have 
$f = f_{\cT_n}$ and 
\begin{align}
    O^{\ox 2} 
    &= \left[(2^n+1)\sum_{\ba}\ketbra{\ba\ba} - \1\right]^{\ox 2}.
\end{align}
Then, we have 
\begin{align}
    \Var_\cE^{(4)}(\rho, \sigma) 
    =& \left(\frac{m-1}{m}\right)^2\tr\left[\cM_\cE^{(4)}(O^{\ox 2}) (\rho\ox\sigma)^{\ox 2}\right]\\
    =& \left(\frac{m-1}{m}\right)^2 \left[(2^n+1)^2 \tr\left[\cM_{\cE}^{(4)}(\Lambda_1)(\rho\ox\sigma)^{\ox 2}\right]
    - 2(2^n+1) \tr\left[\cM_{\cE}^{(2)}(\Lambda_2)(\rho\ox\sigma)\right] + 1\right], 
\end{align}
where $\Lambda_1 := \sum_{\ba,\bb} \ketbra{\ba\ba\bb\bb}$ and $\Lambda_2 := \sum_{\ba} \ketbra{\ba\ba}$.
Then, we compute the above terms one by one. 
With Lemma~\ref{lem:property of unitary 4-design}, we have 
\begin{align}
    \tr\left[\cM_{\cE}^{(4)}(\Lambda_1)(\rho\ox\sigma)^{\ox 2}\right]
    = \frac{(1+\tr[\rho\sigma])^2}{2^n(2^n+1)} + \cO(2^{-3n})
\end{align}
and with Lemma~\ref{lem:properties of unitary design}, 
\begin{align}
    \tr\left[\cM_{\cE}^{(2)}(\Lambda_2)(\rho\ox\sigma)\right]
    = \frac{1}{2^n+1}\left(1+\tr[\rho\sigma]\right). 
\end{align}
Therefore, we have 
\begin{align}
    \Var_\cE^{(1)}(\rho, \sigma) + \Var_\cE^{(4)}(\rho, \sigma) 
    &= \left(\frac{m-1}{m}\right)^2 \left[ 
    \cO(2^{-n}) + \left(1+\frac{1}{2^n}\right) (1 + \tr[\rho\sigma])^2 - 2(1+\tr[\rho\sigma]) + 1
    \right] - \tr^2[\rho\sigma] \\
    &= \left(\frac{m-1}{m}\right)^2 \left[\cO(2^{-n}) + \tr^2[\rho\sigma] + \frac{1}{2^n}(1+\tr[\rho\sigma])^2\right] - \tr^2[\rho\sigma]
    = \cO(2^{-n}). 
\end{align}
\end{proof}

Then, we can propose the sample complexity of DIPE with the unitary $4$-design ensemble with the above lemmas. 
This theorem has been proven in~\cite{anshu2022distributed}, we give it here for completeness and for better comparison.

\begin{theorem}[Theorem 13 of~\cite{anshu2022distributed}]
\label{the:sample complexity of unitary 4-design}
For unknown states $\rho,\sigma$ in $\cH_n$ and the unitary $4$-design ensemble $\cF_n$, the sample complexity of DIPE with $\cF_n$ is $Nm = \Theta(\sqrt{2^n})$. 
\end{theorem}

\begin{proof}[Proof of Theorem~\ref{the:sample complexity of unitary 4-design}]
For any $\varepsilon\in(0, 1)$ and $\delta\in(0,1)$, 
from the Lemma~\ref{lem:variance 2 of unitary 2-design}, Lemma~\ref{lem:variance 3 of unitary 3-design}, and Lemma~\ref{lem:variance 4 of unitary 4-design}, 
it is necessary and sufficient to have
\begin{align}
    &N \geq \frac{1}{\delta\varepsilon^2} \frac{2^n}{m^2}, \quad 
    N \geq \frac{1}{\delta\varepsilon^2} \frac{1}{m}, \quad 
    N \geq \max\left\{\frac{1}{\delta\varepsilon^2} \frac{1}{2^n}, 1\right\}, \\
    \Rightarrow \quad 
    &Nm^2 \geq \frac{2^n}{\delta\varepsilon^2}, \quad 
    Nm \geq \frac{1}{\delta\varepsilon^2}, \quad 
    N \geq \max\left\{\frac{1}{\delta\varepsilon^2 2^n}, 1\right\}. 
\end{align}
Here we ignored constants. 
Therefore, we have 
\begin{align}
    N m 
    \geq N \frac{1}{\sqrt{N}} \sqrt{\frac{2^n}{\delta\varepsilon^2}}
    \geq \max\left\{\frac{1}{\sqrt{\delta\varepsilon^2 2^n}}, 1\right\} 
    \sqrt{\frac{2^n}{\delta\varepsilon^2}}
    = \max\left\{\frac{1}{\delta\varepsilon^2}, \sqrt{\frac{2^n}{\delta\varepsilon^2}}\right\}, 
\end{align}
where the first inequality follows from $Nm^2 \geq 2^n/(\delta\varepsilon^2)$.
Focusing on scalability, we have $Nm = \Theta(\sqrt{2^n})$. 
\end{proof}

\section{Average Sample Complexity}
\label{app:average sample complexity}

Here we discuss the average sample complexity for the following two cases:
\begin{enumerate}
\item $\rho=\sigma = \ketbra{\psi}$, where $\ket{\psi}$ is a Haar random states. 
\item $\rho = \ketbra{\psi}$ and $\sigma = \ketbra{\phi}$, where $\ket{\psi}$ and $\ket{\phi}$ are two independent Haar random states. 
\end{enumerate}
We first consider the sample complexity of DIPE with \emph{Pauli-invariant ensemble} $\cE$ under these two average cases. 

\subsection{Average Variance}

We first prove the sample complexity under average case 1 as follows. 

\begin{proof}[Proof of Theorem~\ref{the:average sample complexity} \textbf{(Average Case 1)}]
We compute the analytic expressions for each of the four variance terms individually.
For the first term of the variance, we have
\begin{align}
\bE_{\psi} \Var_{\cE}^{(1)}(\ketbra{\psi}, \ketbra{\psi}) 
= - \bE_{\psi} \vert\langle\psi\vert\psi\rangle\vert^4
= - 1. 
\end{align}

For the second term of the variance, we have 
\begin{align}
\bE_{\psi} \tr\left[\cM_{\cE}^{(2)} (O^2) \ketbra{\psi}^{\ox 2}\right] 
&= \frac{1}{2^n(2^n+1)} \tr\left[\cM_{\cE}^{(2)} (O^2)\left(\1 + \bigotimes \bS\right)\right] \tag*{Lemma~\ref{lem:haar random states}} \\
&= \frac{1}{2^n(2^n+1)} \left[\sum_{\ba,\bb}f^2_\cE(\ba,\bb) 
+ \sum_{\ba}f_\cE^2(\ba,\ba) \right], \\
\Rightarrow\quad 
\bE_{\psi} \Var_{\cE}^{(2)} (\ketbra{\psi}, \ketbra{\psi}) 
&= \frac{1}{(2^n+1)m^2} \left[f^2_\cE(\bm{0}, \bm{0}) + \norm{f_{\cE}}{2}^2 \right] = \cO\left(\frac{\norm{f_{\cE}}{2}^2}{2^n m^2} \right),  
\tag*{Eq.~\eqref{eq:sum property of pauli-invariant classical function}} 
\end{align}
where 
\begin{align}
\norm{f_{\cE}}{2}^2 := \sum_{\ba}f^2_\cE(\ba,\bm{0}). 
\end{align}

For the third term of the variance, we first have
\begin{align}
&\bE_\psi \bE_{U\sim\cE} \bE_{\ba,\bb,\ba'} f_\cE(\ba,\bb) f_\cE(\ba',\bb) \notag \\
=& \bE_\psi \bE_{U\sim\cE} \sum_{\ba,\bb,\ba'} f_\cE(\ba,\bb) f_\cE(\ba',\bb) \bra{\ba\bb\ba'} U^{\ox 3} \ketbra{\psi}^{\ox 3} 
U^{\dagger\ox 3}\ket{\ba\bb\ba'} \\
=& \sum_{\ba,\bb,\ba'} \frac{f_\cE(\ba,\bb) f_\cE(\ba', \bb)}{2^n(2^n+1)(2^n+2)} 
\left(1 + \delta_{\ba,\bb} + \delta_{\ba,\ba'} + \delta_{\ba',\bb} + 2\delta_{\ba,\bb} \delta_{\ba, \ba'}\right) \tag*{Lemma~\ref{lem:haar random states}} \\
=& \frac{1}{{2^n(2^n+1)(2^n+2)}} \left[\sum_{\ba,\bb,\ba'} f_\cE(\ba,\bb) f_\cE(\ba', \bb) + \sum_{\ba,\bb}f^2_\cE(\ba,\bb) + 2\sum_{\ba,\bb}f_\cE(\ba, \ba) f_{\cE}(\ba,\bb) + 2\sum_{\ba}f^2_\cE (\ba,\ba)\right] \\
=& \frac{1}{{(2^n+1)(2^n+2)}} \left[1 + \norm{f_{\cE}}{2}^2 + 
2 f_\cE(\bm{0}, \bm{0}) + 2 f^2_\cE(\bm{0}, \bm{0})\right]. 
\tag*{Eq.~\eqref{eq:sum property of pauli-invariant classical function}}
\end{align}
Thus, we have
\begin{align}
\bE_{\psi} \Var_{\cE}^{(3)} (\ketbra{\psi}, \ketbra{\psi}) 
&= \frac{2(m-1)}{(2^n+1)(2^n+2)m^2}  \left[1 + \norm{f_{\cE}}{2}^2 + 2 f_\cE(\bm{0}, \bm{0}) + 2 f^2_\cE(\bm{0}, \bm{0})\right] 
= \cO\left(\frac{\norm{f_{\cE}}{2}^2}{4^n m}\right). 
\end{align}

Lastly, for the fourth term of the variance, we have
\begin{align}
&\bE_\psi \tr\left[\cM_{\cE}^{(4)}(O^{\ox 2}) \ketbra{\psi}^{\ox 4}\right] \notag \\
=& \sum_{\ba,\bb,\ba',\bb'} \frac{f_\cE(\ba,\bb) f_\cE(\ba',\bb')}{2^n(2^n+1)(2^n+2)(2^n+3)}  \left(1 + \delta_{\ba,\bb} + \delta_{\ba,\ba'} + \delta_{\ba,\bb'} + \delta_{\bb,\ba'} + \delta_{\bb,\bb'} + \delta_{\ba',\bb'} \right. \notag \\ 
&\left. + 2\delta_{\bb,\ba'}\delta_{\bb,\bb'} + 2\delta_{\ba,\ba'}\delta_{\ba,\bb'} + 
2\delta_{\ba,\bb}\delta_{\ba,\bb'} + 2\delta_{\ba,\bb}\delta_{\ba,\ba'} + 
\delta_{\ba,\bb} \delta_{\ba',\bb'} + \delta_{\ba,\ba'} \delta_{\bb,\bb'} + \delta_{\ba,\bb'}\delta_{\bb, \ba'} + 6 \delta_{\ba,\bb} \delta_{\ba,\ba'} \delta_{\ba,\bb'}\right)
\tag*{Lemma~\ref{lem:haar random states}} \\
=& \frac{1}{(2^n+1)(2^n+2)(2^n+3)} 
\left[2^n + 2^{n+1} f_\cE(\bm{0}, \bm{0}) + 4 + 8 f_\cE(\bm{0}, \bm{0}) + 2^n f_\cE^2(\bm{0}, \bm{0}) + 2\norm{f_{\cE}}{2}^2 + 6 f_\cE^2(\bm{0}, \bm{0})\right] \tag*{Eq.~\eqref{eq:sum property of pauli-invariant classical function}}
\end{align}
Consequently, we have 
\begin{align}
\Var_{\cE}^{(4)} (\ketbra{\psi}, \ketbra{\psi}) 
= \cO\left(\frac{f^2_{\cE}(\bm{0}, \bm{0})}{4^n} + \frac{\norm{f_{\cE}}{2}^2}{8^n}\right). 
\end{align}
Therefore, we have 
\begin{align}
\bE_\psi \Var_\cE(\bX_m) 
= \cO\left(\frac{\norm{f_{\cE}}{2}^2}{2^n m^2} + \frac{\norm{f_{\cE}}{2}^2}{4^n m} + \frac{f^2_{\cE}(\bm{0}, \bm{0})}{4^n} + \frac{\norm{f_{\cE}}{2}^2}{8^n} - 1\right). 
\end{align}
\end{proof}

We now consider the average case 2. 

\begin{proof}[Proof of Theorem~\ref{the:average sample complexity} \textbf{(Average Case 2)}]
We compute the analytic expressions for each of the four variance terms individually.
For the first term of the variance, we have
\begin{align}
\bE_{\psi, \phi} \Var_{\cE}^{(1)} (\ketbra{\psi}, \ketbra{\phi}) 
= - \bE_{\psi}\bE_{\phi} \tr\left[\ketbra{\psi}^{\ox 2} 
\ketbra{\phi}^{\ox 2}\right] 
= - \bE_{\psi} \frac{2}{2^n(2^n+1)}
= - \frac{1}{2^{n-1}(2^n+1)} 
= \cO\left(-\frac{1}{4^n}\right),
\end{align}
where the second equality follows from Lemma~\ref{lem:haar random states}.

For the second term of the variance, we have 
\begin{align}
\bE_{\psi,\phi} \tr\left[\cM_{\cE}^{(2)} (O^2) \ketbra{\psi}\ox \ketbra{\phi}\right] 
&= \frac{1}{4^n} \tr\left[\cM_{\cE}^{(2)} (O^2)\right] 
= \frac{1}{4^n} \sum_{\ba,\bb}f^2_\cE(\ba,\bb), \tag*{Lemma~\ref{lem:haar random states}} \\
\Rightarrow\quad 
\bE_{\psi,\phi} \Var_{\cE}^{(2)} (\ketbra{\psi}, \ketbra{\phi}) 
&= \frac{1}{2^n m^2} \sum_{\ba}f^2_\cE(\ba,\bm{0}) 
= \cO\left(\frac{\norm{f_\cE}{2}^2}{2^n m^2} \right), \tag*{Eq.~\eqref{eq:sum property of pauli-invariant classical function}}
\end{align}

For the third term of the variance, we first have 
\begin{align}
\bE_{\psi, \phi} \bE_{U\sim\cE} \bE_{\ba,\bb,\ba'} f_\cE(\ba,\bb) f_\cE(\ba',\bb) \notag 
=& \bE_{\psi, \phi} \bE_{U\sim\cE} \sum_{\ba,\bb,\ba'} f_\cE(\ba,\bb) f_\cE(\ba',\bb) \bra{\ba\ba'\bb} U^{\ox 3} \ketbra{\psi}^{\ox 2} \ox \ketbra{\phi} U^{\dagger\ox 3} \ket{\ba\ba'\bb} \\
=& \sum_{\ba,\bb,\ba'} \frac{f_\cE(\ba,\bb) f_\cE(\ba', \bb)}{4^n(2^n+1)} 
\left(1 + \delta_{\ba,\ba'} \right) \tag*{Lemma~\ref{lem:haar random states}} \\
=& \frac{1}{4^n(2^n+1)} \left[\sum_{\ba,\bb,\ba'} f_\cE(\ba,\bb) f_\cE(\ba', \bb) + \sum_{\ba,\bb}f^2_\cE(\ba,\bb)\right] \\
=& \frac{1}{2^n(2^n+1)} \left[1 + \norm{f_{\cE}}{2}^2\right]. 
\tag*{Eq.~\eqref{eq:sum property of pauli-invariant classical function}}
\end{align}
Thus, we have 
\begin{align}
\bE_{\psi,\phi} \Var_{\cE}^{(3)} (\ketbra{\psi}, \ketbra{\phi}) 
&= \frac{2(m-1)}{2^n(2^n+1) m^2}  \left[1 + \norm{f_{\cE}}{2}^2\right] 
= \cO\left(\frac{\norm{f_{\cE}}{2}^2}{4^n m}\right). 
\end{align}

Lastly, for the fourth term of the variance, we have
\begin{align}
\bE_{\psi,\phi} \tr\left[\cM_{\cE}^{(4)}(O^{\ox 2}) (\ketbra{\psi}\ox\ketbra{\phi})^{\ox 2}\right] 
=& \sum_{\ba,\bb,\ba',\bb'} \frac{f_\cE(\ba,\bb) f_\cE(\ba',\bb')}{4^n(2^n+1)^2}  \left(1 + \delta_{\ba,\ba'} + \delta_{\bb,\bb'} + \delta_{\ba,\ba'} \delta_{\bb, \bb'}\right)
\tag*{Lemma~\ref{lem:haar random states}} \\
=& \frac{1}{2^n(2^n+1)^2}
\left[2^n + 2 + \norm{f_\cE}{2}^2\right]
\end{align}
Therefore, we have 
\begin{align}
\bE_{\psi,\phi} \Var_\cE(\bX_m) 
= \cO\left(\frac{\norm{f_{\cE}}{2}^2}{2^n m^2} + \frac{\norm{f_{\cE}}{2}^2}{4^n m} + \frac{\norm{f_{\cE}}{2}^2}{8^n}\right). 
\end{align}
\end{proof}

\subsection{Unitary 2-design Ensemble}

\begin{proof}[Proof of Lemma~\ref{lem:average variance of 2-design}]
With the definition of the classical function, it holds for a unitary $2$-design ensemble $\cT_n$ that 
\begin{align}
\norm{f_{\cT_n}}{2}^2 = 4^n + 2^n - 1, \quad 
f_{\cT_n}(\bm{0}, \bm{0}) = 2^n. 
\end{align}
Substituting the expressions into Theorem~\ref{the:average sample complexity} for average case 1, we obtain 
\begin{align}
\Var^{a}_{\cT_n,1} := \bE_\psi \Var_{\cT_n}(\bX_m) 
\approx \frac{2^n + 2}{m^2} + \frac{2}{m} + \frac{1}{2^{n-1}} = \cO\left(\frac{2^n}{m^2} + \frac{1}{m} + \frac{1}{2^n}\right). 
\end{align}
Substituting the expressions into Theorem~\ref{the:average sample complexity} for average case 2, we obtain 
\begin{align}
\Var^{a}_{\cT_n,2} := \bE_{\psi,\phi} \Var_{\cT_n}(\bX_m)
\approx \frac{2^n + 1}{m^2} + \frac{2}{m} + \frac{1}{2^{n-1}} = \cO\left(\frac{2^n}{m^2} + \frac{1}{m} + \frac{1}{2^n}\right). 
\end{align}

We then prove the corresponding average sample complexity as follows. 
For any $\varepsilon\in(0, 1)$ and $\delta\in(0,1)$, 
it is necessary and sufficient to have
\begin{align}
    &N \geq \frac{1}{\delta\varepsilon^2} \frac{2^n}{m^2}, \quad 
    N \geq \frac{1}{\delta\varepsilon^2} \frac{1}{m}, \quad 
    N \geq \max\left\{\frac{1}{\delta\varepsilon^2 2^n}, 1\right\}, \\
    \Rightarrow \quad 
    &Nm^2 \geq \frac{2^n}{\delta\varepsilon^2}, \quad 
    Nm \geq \frac{1}{\delta\varepsilon^2}, \quad 
    N \geq \max\left\{\frac{1}{\delta\varepsilon^2 2^n}, 1\right\}. 
\end{align}
Here we ignored constants. 
Therefore, we have 
\begin{align}
    N m 
    \geq N \frac{1}{\sqrt{N}} \sqrt{\frac{2^n}{\delta\varepsilon^2}}
    \geq \max\left\{\sqrt{\frac{1}{\delta\varepsilon^2 2^n}}, 1\right\} 
    \sqrt{\frac{2^n}{\delta\varepsilon^2}}
    = \max\left\{\frac{1}{\delta\varepsilon^2}, \sqrt{\frac{2^n}{\delta\varepsilon^2}}\right\}, 
\end{align}
where the first inequality follows from $Nm^2 \geq 2^n/(\delta\varepsilon^2)$.
Focusing on scalability, we have $Nm = \cO(\sqrt{2^n})$. 
\end{proof}

\subsection{Local Unitary 2-design Ensemble}

\begin{proof}[Proof of Lemma~\ref{lem:average variance of local 2-design}]
With the definition of the classical function, it holds for a local unitary $2$-design ensemble $\cT_1^{\ox n}$ that  
\begin{align}
\norm{f_{\cT_1^{\ox n}}}{2}^2 = 4^n\left(1 + \frac{1}{4}\right)^n = 5^n, \quad 
f_{\cT_1^{\ox n}}(\bm{0}, \bm{0}) = 2^n. 
\end{align}
Substituting the expressions into Theorem~\ref{the:average sample complexity} for average case 1, we obtain 
\begin{align}
\Var^{a}_{\cT_1^{\ox n},1} := \bE_\psi \Var_{\cT_1^{\ox n}}(\bX_m) 
\approx \frac{2.5^n + 2^n}{m^2} + \frac{2\cdot 1.25^n}{m} + 2\cdot 0.675^n = \cO\left(\frac{2.5^n}{m^2} 
+ \frac{1.25^n}{m} + 0.675^n\right). 
\end{align}
Substituting the expressions into Theorem~\ref{the:average sample complexity} for average case 2, we obtain 
\begin{align}
\Var^{a}_{\cT_1^{\ox n},2} := \bE_{\psi,\phi} \Var_{\cT_1^{\ox n}}(\bX_m)
\approx \frac{2.5^n}{m^2} + \frac{2\cdot 1.25^n}{m} + 2\cdot 0.675^n = \cO\left(\frac{2.5^n}{m^2} 
+ \frac{1.25^n}{m} + 0.675^n\right). 
\end{align}

We then prove the corresponding average sample complexity as follows. 
For any $\varepsilon\in(0, 1)$ and $\delta\in(0,1)$, 
it is necessary and sufficient to have
\begin{align}
    &N \geq \frac{1}{\delta\varepsilon^2} \frac{2.5^n}{m^2}, \quad 
    N \geq \frac{1}{\delta\varepsilon^2} \frac{1.25^n}{m}, \quad 
    N \geq \max\left\{\frac{0.675^n}{\delta\varepsilon^2}, 1\right\}, \\
    \Rightarrow \quad 
    &Nm^2 \geq \frac{2.5^n}{\delta\varepsilon^2}, \quad 
    Nm \geq \frac{1.25^n}{\delta\varepsilon^2}, \quad 
    N \geq \max\left\{\frac{0.675^n}{\delta\varepsilon^2}, 1\right\}. 
\end{align}
Here we ignored constants. 
Therefore, we have 
\begin{align}
    N m 
    \geq N \frac{1}{\sqrt{N}} \sqrt{\frac{2.5^n}{\delta\varepsilon^2}}
    \geq \max\left\{\sqrt{\frac{0.675^n}{\delta\varepsilon^2}}, 1\right\} 
    \sqrt{\frac{2.5^n}{\delta\varepsilon^2}}
    = \max\left\{\frac{1.56^n}{\delta\varepsilon^2}, \sqrt{\frac{2.5^n}{\delta\varepsilon^2}}\right\}, 
\end{align}
where the first inequality follows from $Nm^2 \geq 2.5^n/(\delta\varepsilon^2)$.
Focusing on scalability, we have $Nm = \cO(\sqrt{2.5^n})$. 
\end{proof}

\section{Proof of DIPE with Brickwork Ensemble}
\label{app:DIPE with random brickwork ensemble}

\subsection{Classical Function}
\begin{proof}[Proof of Lemma~\ref{lem:classical function of brikckwork}]
Here we prove the classical function $f_d$. 
Based on Lemma~\ref{lem:requirement for classical function}, for $P,P'\in\cP_n$, we have 
\begin{align}
\tr[\cM_d(O)(P\ox P')] 
&= \bE_{U\sim\cB_d} \tr\left[\left(\sum_{\ba,\bb\in \bZ_2^{n}} f_d(\ba,\bb) U^{\dagger \ox 2} \ketbra{\ba\bb} U^{\ox 2}\right)(P\ox P')
\right] \\
&= \bE_{U\sim\cB_d} \tr\left[\left(\sum_{\ba,\bb\in \bZ_2^{n}} f_d(\ba,\bb) U_{d}^{\dagger \ox 2} \ketbra{\ba\bb} U_d^{\ox 2}\right)(WP W^\dagger\ox W P' W^\dagger)
\right], 
\end{align}
where $U_d$ is the last layer of the circuits and $W$ is the former $d-1$ layers. 
Based on the structure of the last layer, we can only focus on each two-local Clifford gates and have 
\begin{align}
4 \sum_{\ba,\bb\in \bZ_2^{2}} (-2)^{-2\delta_{\ba,\bb}} 
\bE_{V\sim\Cl_2} V^{\dagger\ox 2} \ketbra{\ba\bb} V^{\ox 2} 
= \bigotimes_{i=1}^{2} \bS. 
\end{align}
Thus, each two-local Clifford gate constructs SWAP operators acting on two qubits of two states. 
Therefore, we have 
\begin{align}
\tr[\cM_d(O)(P\ox P')] 
&= \bE_{U\sim\cE} \tr\left[\left(\bigotimes_{i=1}^n \bS\right) (W P W^\dagger\ox W P' W^\dagger)
\right] \\
&= \tr\left[W P W^\dagger W P' W^\dagger\right] \\
&= \tr[PP'] = \begin{cases}
    2^n, & P = P', \\ 
    0, & P\neq P'.
\end{cases}
\end{align}
\end{proof}

\subsection{Average Sample Complexity}

Here, we discuss the average performance guarantee of DIPE with $\cB_d$. 

\begin{proof}[Proof of Lemma~\ref{lem:average variance of the brickwork}] 
With the classical function defined in Lemma~\ref{lem:classical function of brikckwork}, we have $f_d(\bm{0}, \bm{0}) = 2^n$ and 
\begin{align}
\norm{f_d}{2}^2 
= \sum_{\ba} f_d^2(\ba, \bm{0}) 
= 4^n\sum_{\ba}\prod_{s\in S} 16^{-\delta_{\ba_s, \bm{0}}} 
= 4^n\left(1 + \frac{3}{16}\right)^{n/2} = \sqrt{19^n} \approx 4.36^n.
\end{align}
Substituting the expressions into Theorem~\ref{the:average sample complexity} for average case 1, we obtain 
\begin{align}
\Var^{a}_{\cB_d,1} := \bE_\psi \Var_\cE(\bX_m) 
&\approx \frac{4^n + 4.36^n}{2^n m^2} + \frac{2(4.36^n + 2\cdot 4^n + 2\cdot 2^n + 1)}{4^n m} + 2\cdot 0.54^n \\
&\approx \frac{2^n + 2.18^n}{m^2} + \frac{2\cdot 1.09^n}{m} + 2\cdot 0.54^n 
= \cO\left(\frac{2.18^n}{m^2} + \frac{1.09^n}{m} + 0.54^n\right).
\end{align}
Substituting the expressions into Theorem~\ref{the:average sample complexity} for average case 2, we obtain 
\begin{align}
\Var^{a}_{\cB_d,2} := \bE_{\psi,\phi} \Var_\cE(\bX_m) 
&\approx \frac{2.18^n}{m^2} + \frac{2\cdot 1.09^n}{m} + 2\cdot 0.54^n
= \cO\left(\frac{2.18^n}{m^2} + \frac{1.09^n}{m} + 0.54^n\right). 
\end{align}

We then consider the average sample complexity. 
For any $\varepsilon\in(0, 1)$ and $\delta\in(0,1)$, 
it is necessary and sufficient to have
\begin{align}
    &N \geq \frac{1}{\delta\varepsilon^2} \frac{2.18^n}{m^2}, \quad 
    N \geq \frac{1}{\delta\varepsilon^2} \frac{1.09^n}{m}, \quad 
    N \geq \max\left\{\frac{0.54^n}{\delta\varepsilon^2}, 1\right\}, \\
    \Rightarrow \quad 
    &Nm^2 \geq \frac{2.18^n}{\delta\varepsilon^2}, \quad 
    Nm \geq \frac{1.09^n}{\delta\varepsilon^2}, \quad 
    N \geq \max\left\{\frac{0.54^n}{\delta\varepsilon^2}, 1\right\}. 
\end{align}
Here we ignored constants. 
Therefore, we have 
\begin{align}
    N m 
    \geq N \frac{1}{\sqrt{N}} \sqrt{\frac{2.18^n}{\delta\varepsilon^2}}
    \geq \max\left\{\sqrt{\frac{0.54^n}{\delta\varepsilon^2}}, 1\right\} 
    \sqrt{\frac{2.18^n}{\delta\varepsilon^2}}
    = \max\left\{\frac{1.17^n}{\delta\varepsilon^2}, \sqrt{\frac{2.18^n}{\delta\varepsilon^2}}\right\}, 
\end{align}
where the first inequality follows from $Nm^2 \geq 2.18^n/(\delta\varepsilon^2)$.
Focusing on scalability, we have $Nm = \cO(\sqrt{2.18^n})$. 
\end{proof}

\subsection{Asymptotic State-dependent Variance}

Here, we consider the asymptotic state-dependent variance.
With the definition of the second term in Eq.~\eqref{eq:variance 2}, 
we have 
\begin{align}
\Var_{\cB_d}^{(2)} (\rho, \sigma)
&= \frac{1}{m^2} \tr\left[\cM_{\cB_d}^{(2)}\left(\sum_{\ba,\bb\in\bZ_2^n} f_d^2(\ba, \bb)\ketbra{\ba\bb}\right)(\rho\ox\sigma)\right] \\
&= \frac{1}{4^n m^2} \sum_{\ba,\bb\in\bZ_2^n}f_d^2(\ba, \bb) \sum_{P\in\cP_n} \Xi_{\rho,\sigma}(P) \bE_{U\sim\cB_d} \bra{\ba}UPU^\dagger\ket{\ba} \bra{\bb}UPU^\dagger\ket{\bb}, \\
&= \frac{1}{2^n m^2}\sum_{P\in\cP_n} \Xi_{\rho, \sigma}(P) \sum_{\ba\in\bZ_2^n} f_d^2(\ba, \bm{0}) h(\ba, P) 
= \frac{1}{2^n m^2}\sum_{P\in\cP_n} \Xi_{\rho, \sigma}(P) \Upsilon_d(P), 
\end{align}
where we use the property of Pauli-invariant ensemble and $\Xi_{\rho,\sigma}(P) = \tr[P\rho]\tr[P\sigma]$. 
Note that $f_{d}^2(\ba,\bm{0})$ can be represented as a matrix product state (MPS) with $n/2$ tensors $F$, which are $[16, 1, 1,1]$. 

Then, to efficiently compute $h(\ba, P)$, we introduce a \emph{matrix product operator (MPO)} representation of $h(\ba, P)$, as illustrated in Fig.~\ref{fig: MPO representation of coefficient matrix}. 
This construction is based on the physical interpretation of $h(\ba, P)$, as follows, 
\begin{equation}
\begin{split}
    h(\ba, P) 
    =&\bE_U \bra{\bm{0}}U P U^\dagger\ket{\bm{0}} \bra{\ba}U PU^\dagger\ket{\ba} \\
    =& \bE_U \bra{\bm{0}}U P U^\dagger\ket{\bm{0}}\bra{\bm{0}}X^{\ba}UPU^\dagger X^{\ba}\ket{\bm{0}} \\
    =& \Pr_U \left\{U P U^\dagger\in \pm \cZ
    \;\&\; [U P U^\dagger, X^{\ba}] = 0 \right\} - \Pr_U \left\{U P U^\dagger\in \pm \cZ
    \;\&\; \{U P U^\dagger, X^{\ba}\} = 0\right\} \\
    =& \Pr\{UPU^\dagger\in\cZ^{\rm C}_{\ba}\} 
    - \Pr\{UPU^\dagger\in\cZ^{\rm A}_{\ba}\}. 
    \label{eq:physical meaning of h(a,P)} 
\end{split}
\end{equation}

\begin{figure}[!htbp]
\centering
\includegraphics[width=0.9\linewidth]{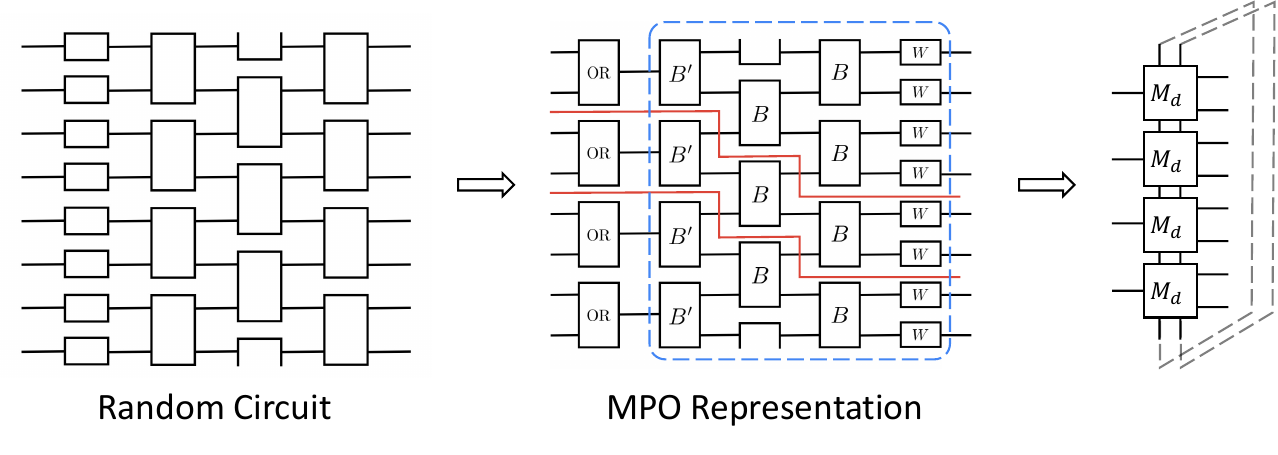}
\caption{
The MPO representation of $h(\ba, P)$.
For each $M_d$, the physical dimension of input leg is $2$, the physical dimension of output leg is $4$, and the bound dimension is $2^{d-1}$.
Therefore, each $M_d$ has $2^{d-1}\times 2^{d-1}$ matrices in $\bR^{4\times 2}$.}
\label{fig: MPO representation of coefficient matrix}
\end{figure}

\begin{proof}[Proof of Lemma~\ref{lem:MPO representation}]
This proof extends Lemma 5 from~\cite{bertoni2024shallow}. We begin by briefly reviewing their approach.

Define $q: \cP_n \to \bZ_2^n$ as the \emph{signature} of the Pauli operator $P$,  
\begin{align}
    [q(P)]_i = \begin{cases}
        0, & P_i = I \\
        1, & \text{otherwise}, 
    \end{cases}
    \label{eq:signature of P}
\end{align}
The effect of each two-qubit Clifford gate can be represented by the matrix
\begin{align}
    B := \left[\begin{matrix}
        1 & 0 & 0 & 0 \\
        0 & 0.2 & 0.2 & 0.2 \\
        0 & 0.2 & 0.2 & 0.2 \\
        0 & 0.6 & 0.6 & 0.6
    \end{matrix}\right]. 
\end{align}
The physical interpretation of $B$ is as follows:
\begin{itemize}
    \item If the input is $00$, the output is deterministically $00$.
    \item Otherwise, the outputs are $01$, $10$, or $11$ with probabilities $0.2$, $0.2$, and $0.6$, respectively.
\end{itemize}

Applying $d$ layers of such $B$ matrices forms a tensor network, with input legs labeled by $q(P)\in \bZ_2^n$ and output legs by $\bm{\gamma}\in \bZ_2^n$. The resulting tensor evaluates the probability
\begin{align}
    \Pr\left\{q(UPU^\dagger)=\bm{\gamma}\right\}. 
\end{align}
We then multiply this by the conditional probability
\begin{align}
    \Pr\left\{UPU^\dagger\in\pm\cZ\;\vert\; q(UPU^\dagger)=\bm{\gamma}\right\},
\end{align}
and sum over $\bm{\gamma}$ to obtain:
\begin{align}
    h(\bm{0},P)
    =\Pr\left\{UPU^\dagger\in\pm\cZ\right\}
    =\sum_{\bm{\gamma}} 
    \Pr\left\{q(UPU^\dagger)=\bm{\gamma}\right\}\cdot
    \Pr\left\{UPU^\dagger\in\pm\cZ\;\vert\; q(UPU^\dagger)=\bm{\gamma}\right\}.
\end{align}
The condition $UPU^\dagger \in \pm\cZ$ can be checked locally: each local Pauli must be either $I$ or $Z$. 
For $\bm{\gamma}_i = 1$, the output is $Z$ with probability $1/3$. 
Therefore, the total contraction involves applying a weight vector $W_0 = [1, 1/3]$ on each output leg, 
where the entry reflects whether the local signature is $0$ or $1$.

We now generalize this to represent $h(\ba, P)$ using an MPO. 
Recall from Eq.~\eqref{eq:physical meaning of h(a,P)} that:
\begin{align}
    h(\ba, P) 
    &= \Pr_U \left\{U P U^\dagger\in \pm \cZ
    \;\&\; [U P U^\dagger, X^{\ba}] = 0 \right\} - 
    \Pr_U \left\{U P U^\dagger\in \pm \cZ
    \;\&\; \{U P U^\dagger, X^{\ba}\} = 0\right\} \\
    &= \sum_{\bm{\gamma}} \Pr\left\{q(UPU^\dagger)=\bm{\gamma}\right\}\cdot
    \Pr\left\{UPU^\dagger\in\pm\cZ\;\vert\; q(UPU^\dagger)=\bm{\gamma}\right\}\cdot 
    (-1)^{\vert{\rm supp}(\bm{\gamma}\oplus \ba)\vert}.
\end{align}
To encode the sign $(-1)^{\left|\mathrm{supp}(\bm{\gamma} \oplus \ba)\right|}$, we define another vector $W_1 = [1, -1/3]$. 
For each site $i$, if $\ba_i = j$, we apply $W_j$ to $\bm{\gamma}_i$, where $j = 0,1$. 
Hence, we can compute $h(\ba, P)$ by setting the left input leg to $q(P)$ and the right leg to $\ba$.

Finally, we simplify the network to illustrate its scalability, as done in~\cite{bertoni2024shallow}. Noting that the last three columns of $B$ are identical, we define the reduced tensor
\begin{align}
    B' = \left[\begin{matrix}
        1 & 0 \\ 0 & 0.2 \\ 0 & 0.2 \\ 0 & 0.2
    \end{matrix}\right]. 
\end{align}
A $1$-depth circuit has bond dimension $1$ as it only consists of the tensor $B'$. 
Each additional layer doubles the bond dimension, leading to a final bond dimension of $2^{d-1}$. 

Therefore, with the MPS representation of $f_d^2(\ba, \bm{0})$ and the MPO representation of $h(\ba, P)$, we have 
\begin{align}
\Upsilon_d(P) = \sum_{\ba} f_d^2(\ba, \bm{0}) h(\ba, P) 
= \tr\left[\prod_{i=1}^{n/2} \sum_{a} F_a\cdot M_d^{a, P_i}\right], 
\end{align}
which can be computed efficiently. 
\end{proof}

\section{Proof of the State-Dependent Variances of DIPE with Clifford Ensembles}
\label{app:clifford}

We provide a detailed analysis of the state-dependent variances of DIPE with global and local Clifford ensembles.

\subsection{Global Clifford}

Since the global Clifford ensemble is a unitary $3$-design ensemble, 
we have
\begin{align}
    f_{\Cl_n} (\ba,\bb) = f_{\cT_n} (\ba,\bb) = 
    \begin{cases}
        2^n, & \ba = \bb, \\
        -1, & \ba \neq \bb. 
    \end{cases}
\end{align}
Additionally, with Lemma~\ref{lem:variance 2 of unitary 2-design} and Lemma~\ref{lem:variance 3 of unitary 3-design}
we have 
\begin{align}
    \Var_{\Cl_n}^{(2)}(\rho, \sigma) &= \cO\left(\frac{2^n}{m^2}\right), \quad 
    \Var_{\Cl_n}^{(3)}(\rho, \sigma) = \cO\left(\frac{1}{m}\right). 
\end{align}
Although the global Clifford ensemble fails to be a unitary $4$-design, we can also compute the $4$-moment of $\Cl_n$ with Schur-Weyl duality theory for the Clifford group~\cite{gross2021schur,chen2024nonstabilizerness}, 
the result is shown in Theorem~\ref{the:global clifford variance}. 
Here, we provide the proof. 
\begin{proof}[Proof of Theorem~\ref{the:global clifford variance}]
As shown in the proof of Lemma~\ref{lem:variance 4 of unitary 4-design}, we can decompose $\Var_{\Cl_n}^{(4)}$ as 
\begin{align}
    \Var_{\Cl_n}^{(4)}(\rho, \sigma) 
    &= \left(\frac{m-1}{m}\right)^2 
    \left[(2^n+1)^2\tr\left[\cM_{\Cl_n}^{(4)}(\Lambda_1)(\rho\ox \sigma)\right] - 2\tr[\rho\sigma] - 1\right].
\end{align}
With Lemma~\ref{lem:4-moment of global clifford} proved below, we have 
\begin{align}
    \Var_{\Cl_n}^{(1)}(\rho, \sigma) + \Var_{\Cl_n}^{(4)}(\rho, \sigma) 
    &\leq \left(\frac{m-1}{m}\right)^2 \left[\left(1-\frac{1}{2^n+2}\right)
    \left[(1+\tr[\rho\sigma])^2 + \frac{1}{2^{n-1}} \norm{\Xi_{\rho,\sigma}}{2}^2\right] - 2\tr[\rho\sigma] - 1\right] - \tr^2[\rho\sigma] \\
    &= \cO(2^{-n}) + \left(\frac{m-1}{m}\right)^2\frac{2^n+1}{2^{n-1}(2^n+2)} \norm{\Xi_{\rho,\sigma}}{2}^2, 
\end{align}
where 
\begin{align}
    \norm{\Xi_{\rho,\sigma}}{2}^2 := \sum_{P\in\cP_n} \tr^2[P\rho] \tr^2[P\sigma].
\end{align}
Using the fact that $\tr^2[P\rho]\leq 1$ and $\sum_P \tr^2[P\rho] \leq 2^n$ for all state $\rho$, we have 
\begin{align}
    \norm{\Xi_{\rho,\sigma}}{2}^2 \leq 2^n. 
\end{align}
Therefore, we have
\begin{align}
    \Var_{\Cl_n}^{(1)}(\rho, \sigma) + \Var_{\Cl_n}^{(4)}(\rho, \sigma) 
    \leq \cO(1) . 
\end{align}

Then, we can propose the sample complexity of DIPE with the global Clifford ensemble.
For any $\varepsilon\in(0, 1)$ and $\delta\in(0,1)$, 
it is necessary and sufficient to have
\begin{align}
&N \geq \frac{1}{\delta\varepsilon^2} \frac{2^n}{m^2}, \quad 
N \geq \frac{1}{\delta\varepsilon^2} \frac{1}{m}, \quad 
N \geq  \frac{1}{\delta\varepsilon^2}, \\
\Rightarrow \quad 
&Nm^2 \geq \frac{2^n}{\delta\varepsilon^2}, \quad 
Nm \geq \frac{1}{\delta\varepsilon^2}, \quad 
N \geq \frac{1}{\delta\varepsilon^2}. 
\end{align}
Here we ignored constants. 
Therefore, we have 
\begin{align}
    N m 
    \geq N \frac{1}{\sqrt{N}} \sqrt{\frac{2^n}{\delta\varepsilon^2}}
    \geq \frac{1}{\sqrt{\delta\varepsilon^2}}
    \sqrt{\frac{2^n}{\delta\varepsilon^2}}
    = \frac{\sqrt{2^n}}{\delta\varepsilon^2}. 
\end{align}
Focusing primarily on scalability, we have $Nm = \Theta(\sqrt{2^n})$. 
\end{proof}

\subsection{Local Clifford}

When the unitary ensemble is the local Clifford ensemble $\Cl_1^{\ox n}$, we have 
\begin{align}
    f_{\Cl_1^{\ox n}}(\ba,\bb) 
    = f_{\cT_1^{\ox n}}(\ba,\bb) 
    = 2^n\cdot(-2)^{-\cD(\ba,\bb)}. 
\end{align}

Then, we consider the state-dependent variance and focus on the second term and the fourth term. 
We prove Theorem~\ref{the:local clifford variance} as follows. 
\begin{proof}[Proof of Theorem~\ref{the:local clifford variance}]
With the definition of the classical function, we have 
\begin{align}
    O^2 = \bigotimes_{i=1}^n \left(4\ketbra{00} + 4\ketbra{11} + \ketbra{01} + \ketbra{10}\right) = \bigotimes_{i=1}^n O', 
\end{align}
where $O' = 4\ketbra{00} + 4\ketbra{11} + \ketbra{01} + \ketbra{10}$. 
Thus, we have 
\begin{align}
    \tr[O'] = 10, \quad
    \tr[\bS O'] = 8. 
\end{align}
Then, with the property of unitary $2$-design shown in Lemma~\ref{lem:unitary 2-design expand}, we have 
\begin{align}
    \Var_{\Cl_1^{\ox n}}^{(2)}(\rho,\sigma)
    &= \frac{1}{m^2} \tr\left[\bigotimes(2\1 + \bS) (\rho\ox\sigma)\right]. 
\end{align}
With the decomposition of SWAP operator $\bS = \sum_{P\in\cP_1} P^{\ox 2} / 2$, we have 
\begin{align}
    \tr\left[\bigotimes(2\1 + \bS) (\rho\ox\sigma)\right]
    &= \frac{1}{2^n}\tr\left[\bigotimes\left(5I^{\ox 2} + X^{\ox 2} + Y^{\ox 2} + Z^{\ox 2}\right) (\rho\ox\sigma)\right] \\
    &= \frac{1}{2^n} \sum_{P\in\cP_n} 5^{n-|P|}\tr[P^{\ox 2} (\rho\ox\sigma)] \\
    &= \left(\frac{5}{2}\right)^n \sum_{P\in\cP_n} 5^{-|P|} \Xi_{\rho,\sigma}(P), 
\end{align}
where $\Xi_{\rho,\sigma}(P) := \tr[P\rho]\tr[P\sigma]$. 
Then, we consider bounding the sum over $P\in\cP_n$. 
Using the Cauchy-Schwarz inequality, we have 
\begin{align}
    \sum_{P\in\cP_n} 5^{-|P|} \Xi_{\rho,\sigma}(P)
    \leq \sqrt{\sum_{P\in\cP_n} 5^{-|P|} \tr^2[P\rho]} 
    \sqrt{\sum_{P\in\cP_n} 5^{-|P|} \tr^2[P\sigma]} 
    \leq \left(\frac{6}{5}\right)^n, 
\end{align}
where we use Lemma~\ref{lem: bound for sum over P (local Clifford)}. 
Therefore, the second term of the variance has the following bound, 
\begin{align}
    \Var_{\Cl_1^{\ox n}}^{(2)}(\rho, \sigma) 
    \leq \frac{3^n}{m^2}, 
\end{align}
where the upper bound is achieved when $\rho=\sigma$ is a product state. 
Now we compute the fourth term of the variance. 
$O$ can also be written as 
\begin{align}
O 
&= \bigotimes \left(\frac{1}{2}I^{\ox 2} + \frac{3}{2}Z^{\ox 2}\right), \quad
O^{\ox 2} 
= \left(\frac{1}{4}\right)^{n} \bigotimes \left(I^{\ox 4} + 3I^{\ox 2}\ox Z^{\ox 2} + 3Z^{\ox 2}\ox I^{\ox 2} + 3^2 Z^{\ox 4}\right). 
\end{align}
\begin{align}
\cM_{\Cl_1^{\ox n}}^{(4)}(O^{\ox 2}) 
&= \left(\frac{1}{4}\right)^{n} \bigotimes_{i=1}^n \bE_{U_i\sim\Cl_1}
U_i^{\dagger\ox4} \left(I^{\ox 4} + 3 I^{\ox 2}\ox Z^{\ox 2} + 3Z^{\ox 2}\ox I^{\ox 2} + 3^2 Z^{\ox 4}\right) U_i^{\ox 4} \\
&= \left(\frac{1}{4}\right)^{n} \bigotimes_{i=1}^n \left[I^{\ox 4} + 3 I^{\ox 2}\ox \cM_{\Cl_1}^{(2)}(Z^{\ox 2}) + 3 \cM_{\Cl_1}^{(2)}(Z^{\ox 2}) \ox I^{\ox 2} + 3^2 \cM_{\Cl_1}^{(4)} (Z^{\ox 4})\right] \\
&= \left(\frac{1}{4}\right)^{n} \bigotimes_{i=1}^n \left(I^{\ox 4} + 3I^{\ox 2}\ox \bF^{(2)} + 3 \bF^{(2)} \ox I^{\ox 2} + 3^2 \bF^{(4)}\right), 
\end{align}
where we use Lemma~2 of~\cite{zhou2023performance} and 
\begin{align}
    \bF^{(k)} := \frac{1}{3}\left(X^{\ox k} + Y^{\ox k} + Z^{\ox k}\right). 
\end{align}
For $P_i\in\cP_1$, $i = 1, 2, 3, 4$, we have 
\begin{align}
\tr\left[\left(I^{\ox 4} + 3I^{\ox 2}\ox \bF^{(2)} + 3 \bF^{(2)} \ox I^{\ox 2} + 3^2 \bF^{(4)}\right)P_1\ox P_2\ox P_3 \ox P_4\right] 
= \begin{cases}
    16, & P_1 = P_2 = P_3 = P_4 = I, \\
    16, & P_1 = P_2 = I, P_3 = P_4 \neq I, \\
    16, & P_1 = P_2 \neq I, P_3 = P_4 = I, \\
    16\cdot 3, & P_1 = P_2 = P_3 = P_4 \neq I, \\
    0, & \text{else}. 
\end{cases}
\end{align}
Therefore, we have 
\begin{align}
    \tr\left[\cM_{\Cl_1^{\ox n}}^{(4)}(O^{\ox 2}) (\rho\ox\sigma)^{\ox 2}\right]
    &= \frac{1}{16^n} \sum_{P, Q\in\cP_n} \Xi_{\rho,\sigma}(P) \Xi_{\rho,\sigma}(Q) \tr\left[\cM_{\Cl_1^{\ox n}}^{(4)}(O^{\ox 2}) (P\ox P\ox Q\ox Q)\right] \\
    &= \frac{1}{4^n} \sum_{P} \sum_{Q\in \cP_n^P} \frac{\Xi_{\rho,\sigma}(P) \Xi_{\rho,\sigma}(Q)}{3^{-|\{i|P_i=Q_i\neq I\}|}} , 
\end{align}
where 
\begin{align}
    \cP_n^P := \left\{Q\in\cP_n | \forall i, |P_i| \cdot |Q_i| = 0 \text{ or } P_i = Q_i\right\}
\end{align}
and $|\cP_n^P| = 2^{2n - |P|}$. 
Specially, if $\rho=\sigma$ is a stabilizer product state, we have 
\begin{align}
\tr\left[\cM_{\Cl_1^{\ox n}}^{(4)}(O^{\ox 2})\rho^{\ox 4}\right] = 1.5^n. 
\end{align}

We then consider the sample complexity when $\rho = \sigma$ is a \emph{stabilizer product state}. 
We have 
\begin{align}
\Var_{\Cl_1^{\ox n}}^{(2)}(\rho, \sigma) = \frac{3^n}{m^2}, \quad 
\Var_{\Cl_1^{\ox n}}^{(4)}(\rho, \sigma) = 1.5^n \left(\frac{m-1}{m}\right)^2. 
\end{align}
Therefore, for any $\varepsilon\in(0,1)$ and $\delta\in(0,1)$, it is necessary and sufficient to have 
\begin{align}
&N \geq \frac{1}{\delta\varepsilon^2} \frac{3^n}{m^2}, \quad 
N \geq  \frac{1}{\delta\varepsilon^2}, \\
\Rightarrow \quad 
&Nm^2 \geq \frac{3^n}{\delta\varepsilon^2}, \quad 
N \geq \frac{1.5^n}{\delta\varepsilon^2}. 
\end{align}
Here we ignored constants. 
Therefore, we have 
\begin{align}
    N m 
    \geq N \frac{1}{\sqrt{N}} \sqrt{\frac{3^n}{\delta\varepsilon^2}}
    \geq \sqrt{\frac{1.5^n}{\delta\varepsilon^2}}
    \sqrt{\frac{3^n}{\delta\varepsilon^2}}
    = \frac{\sqrt{4.5^n}}{\delta\varepsilon^2}. 
\end{align}
Therefore, we have $Nm = \cO(\sqrt{4.5^n})$. 
\end{proof}

\begin{lemma}
\label{lem: bound for sum over P (local Clifford)}
For arbitrary $n$-qubit state $\rho$, we have 
\begin{align}
\sum_{P\in\cP_n} 5^{-|P|} \tr^2[P\rho] \leq \left(\frac{6}{5}\right)^n. 
\end{align}
\end{lemma}
\begin{proof}
Obviously, the maximum of L.H.S is achieved when $\rho$ is a pure state. 
Thus, we only consider pure state. 
Firstly, for $n=1$ case, we must have
\begin{align}
\sum_{P\in\cP_1} 5^{-|P|} \tr^2[P\rho] = \left(\frac{6}{5}\right)^1.
\end{align}
Then, suppose that this inequality is held for $n-1$ case. 
For the $n$-qubit case, with Schmidt decomposition, 
we can decompose an $n$-qubit state $\rho = \ketbra{\psi}$ as 
\begin{align}
\ket{\psi} &= a_0\ket{\varphi_0 \phi_0} + a_1\ket{\varphi_1 \phi_1}, \quad
\rho = \sum_{i,j=0}^1 a_ia_j \ket{\varphi_i \phi_i}\!\bra{\varphi_j\phi_j}, 
\end{align}
where $a_0^2 + a_1^2 = 1$ and $\ket{\varphi_0}, \ket{\varphi_1}$ are single-qubit states. 
Then, for each $P\in\cP_n$, we have 
\begin{align}
\frac{1}{5^{|P|}} \tr^2[P\rho] 
&= \frac{1}{5^{|P|}} \left(\sum_{i,j=0}^1 a_i a_j \bra{\varphi_j \phi_j} P \ket{\varphi_i \phi_i} \right)^2 \\
&= \frac{1}{5^{|P|}} \sum_{i,j,k,l=0}^1 a_ia_j a_k a_l
\bra{\varphi_j \phi_j \varphi_k \phi_k}  P^{\ox 2} \ket{\varphi_i \phi_i \varphi_l \phi_l}. 
\end{align}
For each $\bra{\varphi_j\phi_j \varphi_k\phi_k}  P^{\ox 2} \ket{\varphi_i\phi_i \varphi_l\phi_l}$, 
we compute its sum over $P$ as follows, 
\begin{align}
\sum_{P\in\cP_n} \frac{1}{5^{|P|}}
\bra{\varphi_j\phi_j \varphi_k\phi_k}  P^{\ox 2} \ket{\varphi_i\phi_i \varphi_l\phi_l}
&= \left(\sum_{P'\in\cP_1} \frac{1}{5^{|P'|}} \bra{\varphi_j\varphi_k}P'^{\ox 2}\ket{\varphi_i\varphi_l}\right) 
\left(\sum_{P''\in\cP_{n-1}} \frac{1}{5^{|P''|}} \bra{\phi_j \phi_k}P''^{\ox 2}\ket{\phi_i\phi_l}\right) \\
&= \left(\frac{4}{5}\delta_{ij}\delta_{kl} + \frac{2}{5}\bra{\varphi_j\varphi_k}\bS\ket{\varphi_i\varphi_l} \right) 
\left(\sum_{P''\in\cP_{n-1}} \frac{1}{5^{|P''|}} \bra{\phi_j \phi_k}P''^{\ox 2}\ket{\phi_i\phi_l}\right) \\
&= \left(\frac{4}{5}\delta_{ij}\delta_{kl} + \frac{2}{5}\delta_{ik} \delta_{jl} \right)
\left(\sum_{P''\in\cP_{n-1}} \frac{1}{5^{|P_2|}} \bra{\phi_j \phi_k}P''^{\ox 2}\ket{\phi_i\phi_l}\right) \\
&\leq \frac{6}{5} \sum_{P''\in\cP_{n-1}} \frac{1}{5^{|P''|}} \bra{\phi_j \phi_k}P''^{\ox 2}\ket{\phi_i\phi_l}. 
\end{align}
Therefore, we have 
\begin{align}
\sum_{P\in\cP_n} \frac{1}{5^{|P|}} \tr^2[P\rho] 
&\leq \frac{6}{5} \sum_{P''\in\cP_{n-1}} \sum_{i,j,k,l=0}^1 
\frac{a_i a_j a_k a_l}{5^{|P''|}} \bra{\phi_j \phi_k}P''^{\ox 2}\ket{\phi_i\phi_l} \\
&= \frac{6}{5} \sum_{P''\in\cP_{n-1}} \frac{1}{5^{|P''|}} 
\tr^2\left[P''\left(\sum_{i,j} a_i a_j\ket{\phi_i}\!\bra{\phi_j} \right)\right] \\
&= \frac{6}{5} \sum_{P''\in\cP_{n-1}} \frac{1}{5^{|P''|}} 
\tr^2[P''\sigma] \leq \left(\frac{6}{5}\right)^n, 
\end{align}
where we define $\sigma := \sum_{i,j} a_i a_j\ket{\phi_i}\!\bra{\phi_j}$ and the last inequality uses our assumption. 
\end{proof}

\section{Performance Guarantee of DIPE with Approximate Unitary Design Ensembles}
\label{app:approximate design}

We bound the bias of the estimator as follows,
\begin{align}
|\tilde{\omega} - \tr[\rho\sigma]| 
&= \left|\tr\left[\cM_{\tilde{\cF}_n}^{(2)}(O)(\rho\ox\sigma)\right] - \tr\left[\cM_{\cF_n}^{(2)}(O)(\rho\ox\sigma)\right] \right| \\
&= \left|\tr\left[\left(\cM_{\tilde{\cF}_n}^{(2)} - \cM_{\cF_n}^{(2)}\right) \left((2^n+1)\sum_{\ba} \ketbra{\ba\ba} - \1\right) (\rho\ox\sigma)\right] \right| \\
&= (2^n+1) \left|\tr\left[\left(\cM_{\tilde{\cF}_n}^{(2)} - \cM_{\cF_n}^{(2)}\right) \left(\sum_{\ba} \ketbra{\ba\ba}\right) (\rho\ox\sigma)\right] \right| \\
&\leq \epsilon (2^n+1) \tr\left[\cM_{\cF_n}^{(2)}\left(\sum_{\ba} \ketbra{\ba\ba}\right) (\rho\ox \sigma)\right] \\
&= \epsilon \left(1 + \tr[\rho\sigma]\right) \leq 2\epsilon, 
\end{align}
where we use the definition of approximate unitary $4$-design and the property of unitary $2$-design.
We now turn to the variance, which is defined in Eq.~\eqref{eq:variance appendix}, 
\begin{align}
\Var[\tilde{\bX}_m]
=& \frac{1}{m^2} \tr\left[\cM_{\tilde{\cF}_n}^{(2)}(O^2) (\rho\ox\sigma)\right] + \left(\frac{m-1}{m}\right)^2\tr\left[\cM_{\tilde{\cF}_n}^{(4)}(O^{\ox 2}) (\rho\ox\sigma)^{\ox 2}\right] - \left[\bE_{U\sim\tilde{\cF}_n,\ba,\bb} f_{\cT_n}(\ba,\bb)\right]^2 \notag \\
+& \frac{m-1}{m^2} \bE_{U\sim\tilde{\cF}_n} \left[\bE_{\ba,\bb} f_{\cT_n}(\ba,\bb) \left(\bE_{\ba'} f_{\cT_n}(\ba',\bb) + \bE_{\bb'} f_{\cT_n}(\ba,\bb')\right)\right].
\end{align}
In the following, we bound each term one by one.
\begin{enumerate}
\item For the first term, we have
\begin{align}
\tr\left[\cM_{\tilde{\cF}_n}^{(2)}(O^2) (\rho\ox\sigma)\right]
&\leq (1 + \epsilon) \tr\left[\cM_{\cF_n}^{(2)}(O^2) (\rho\ox\sigma)\right], 
\end{align}
with the fact that $O^2$ is a positive semi-definite operator. 
Thus, with Lemma~\ref{lem:properties of unitary design}, we have
\begin{align}
\tr\left[\left(\cM_{\tilde{\cF}_n^{(4)}} - \cM_{\cF_4}^{(4)}(\rho\ox\sigma)\right)\right] 
&\leq \epsilon \tr\left[\cM_{\cF_n}^{(2)}(O^2) (\rho\ox\sigma)\right] \\
&= \cO(\epsilon\cdot 2^n). 
\end{align}
\item For the second term, we have
\begin{align}
\tr\left[\cM_{\tilde{\cF}_n}^{(4)}(O^{\ox 2}) (\rho\ox\sigma)^{\ox 2}\right]
&= (2^n+1)^2 \tr\left[\cM_{\tilde{\cF}_n}^{(4)}(\Lambda_1)(\rho\ox\sigma)^{\ox 2}\right]
- 2(2^n+1) \tr\left[\cM_{\tilde{\cF}_n}^{(2)}(\Lambda_2)(\rho\ox\sigma)\right] + 1 \\
&\leq (1+\epsilon)(2^n+1)^2 \tr\left[\cM_{\cF_n}^{(4)}(\Lambda_1)(\rho\ox\sigma)^{\ox 2}\right] - 2(1-\epsilon)(2^n+1) \tr\left[\cM_{\cF_n}^{(2)}(\Lambda_2)(\rho\ox\sigma)\right] + 1, \notag 
\end{align}
where $\Lambda_1 := \sum_{\ba,\bb} \ketbra{\ba\ba\bb\bb}$ and $\Lambda_2 := \sum_{\ba} \ketbra{\ba\ba}$. 
Here we use the definition of approximate unitary $4$-design.
Thus, with Lemma~\ref{lem:properties of unitary design}, we have
\begin{align}
\tr\left[\left(\cM_{\tilde{\cF}_n}^{(4)} - \cM_{\cF_n}^{(4)}\right)(O^{\ox 2}) (\rho\ox\sigma)^{\ox 2}\right]
&\leq \epsilon (2^n + 1)^2 \tr\left[\cM_{\cF_n}^{(4)}(\Lambda_1)(\rho\ox\sigma)^{\ox 2}\right] + 2\epsilon (2^n+1) \tr\left[\cM_{\cF_n}^{(2)}(\Lambda_2)(\rho\ox\sigma)\right] \notag \\
&= \cO(\epsilon). 
\end{align}
\item For the third term, since $|\tilde{\omega} - \tr[\rho\sigma]| \leq \epsilon(1 + \tr[\rho\sigma])$, we have
\begin{align}
\left[\bE_{U\sim\tilde{\cF}_n,\ba,\bb} f_{\cT_n}(\ba,\bb)\right]^2
&\geq \left[(1 - \epsilon)\tr[\rho\sigma] - \epsilon\right]^2 
= (1-\epsilon)^2 \tr^2[\rho\sigma] - 2\epsilon(1-\epsilon)\tr[\rho\sigma] + \epsilon^2.
\end{align}
Thus, we have 
\begin{align}
- \left[\bE_{U\sim\tilde{\cF}_n,\ba,\bb} f_{\cT_n}(\ba,\bb)\right]^2 + \tr^2[\rho\sigma]
&\geq \epsilon(2-\epsilon)\tr^2[\rho\sigma] + 2\epsilon(1-\epsilon)\tr[\rho\sigma] - \epsilon^2 \\
&= \cO(\epsilon).
\end{align}
\item For the fourth term, there are two similar terms. 
We bound one of them as follows, 
\begin{align}
&\bE_{U\sim\tilde{\cF}_n} \left[\bE_{\ba,\bb} f_{\cT_n}(\ba,\bb) \bE_{\ba'} f_{\cT_n}(\ba',\bb)\right] \notag \\
=& (2^n+1)^2 \tr\left[\cM_{\tilde{\cF}_n}^{(4)}(\Lambda_3)(\rho\ox\rho\ox \sigma)\right] 
- \tr\left[\cM_{\tilde{\cF}_n}^{(2)}(\Lambda_2)(\rho\ox\sigma)\right] 
- \tr\left[\cM_{\tilde{\cF}_n}^{(2)}(O)(\rho\ox\sigma)\right] \\
\leq& (1+\epsilon)(2^n+1)^2 \tr\left[\cM_{\cF_n}^{(4)}(\Lambda_3)(\rho\ox\rho\ox \sigma)\right] 
- (1-\epsilon) \tr\left[\cM_{\cF_n}^{(2)}(\Lambda_2)(\rho\ox\sigma)\right] 
- \tr\left[\cM_{\tilde{\cF}_n}^{(2)}(O)(\rho\ox\sigma)\right], 
\end{align}
where $\Lambda_3 := \sum_{\ba} \ketbra{\ba\ba\ba}$. 
Thus, with Lemma~\ref{lem:properties of unitary design}, we have  
\begin{align}
&\bE_{U\sim\tilde{\cF}_n} \left[\bE_{\ba,\bb} f_{\cT_n}(\ba,\bb) \bE_{\ba'} f_{\cT_n}(\ba',\bb)\right] - \bE_{U\sim\cF_n} \left[\bE_{\ba,\bb} f_{\cT_n}(\ba,\bb) \bE_{\ba'} f_{\cT_n}(\ba',\bb)\right] \notag \\
\leq& \epsilon (2^n + 1)^2 \tr\left[\cM_{\cF_n}^{(4)}(\Lambda_3)(\rho\ox\rho\ox\sigma)\right] + \epsilon \tr\left[\cM_{\cF_n}^{(2)}(\Lambda_2)(\rho\ox\sigma)\right] + 2\epsilon \\
\leq& \cO(\epsilon).
\end{align}
A similar bound holds for the other term, i.e., 
\begin{align}
&\bE_{U\sim\tilde{\cF}_n} \left[\bE_{\ba,\bb} f_{\cT_n}(\ba,\bb) \bE_{\bb'} f_{\cT_n}(\ba,\bb')\right] - \bE_{U\sim\cF_n} \left[\bE_{\ba,\bb} f_{\cT_n}(\ba,\bb) \bE_{\bb'} f_{\cT_n}(\ba,\bb')\right] 
\leq \cO(\epsilon).
\end{align}
\end{enumerate}
Now, we combine all the bounds above and have
\begin{align}
\Var[\tilde{\bX}_m] - \Var_{\cF_n}[\bX_m]
&\leq \cO\left(\frac{\epsilon\cdot 2^n}{m^2} + \epsilon \right)
\end{align}

\section{Additional Numeric}
\label{app:more experiments}

\subsection{Haar random states}

Here, we compare the performance of $\Cl_1^{\ox n}$, $\cB_d$ ($d=1,3$), and $\Cl_n$,
across systems ranging from $4$ to $26$ qubits in steps of $2$. 
For each $n$, we generate $10^2$ pairs of Haar random states $\{\ket{\psi_i}, \ket{\phi_i}\}_{i=1}^{N_s}$. 
For each pair, we sample $10^2$ unitaries from each of the three ensembles and estimate their inner product using $m = 10^1, 10^2, 10^3$ measurement shots.
Thus, for each state pair, we obtain $10^2$ independent estimators of their inner product. 
We compute the variance of these estimators for each ensemble, and the results are shown in Fig.~\ref{fig: experiment of haar random states}.
To benchmark performance, where reference lines corresponding to the average variance. 
Specifically, for the local and global Clifford ensembles, the predicted average variances scale as $2.5^n/m^2$ and $2^n/m^2$, respectively.
For the $d$-depth brickwork ensemble, the predicted average variance is $2.18^n/m^2$, as given in Lemma~\ref{lem:average variance of the brickwork}.
We observe excellent agreement between these theoretical predictions and our numerical results. 
Additionally, we find that the depth of brickwork ensemble does not affect the average sample complexity when two states are independent random Haar states. 

\begin{figure}[!htbp]
\centering
\includegraphics[width=0.7\linewidth]{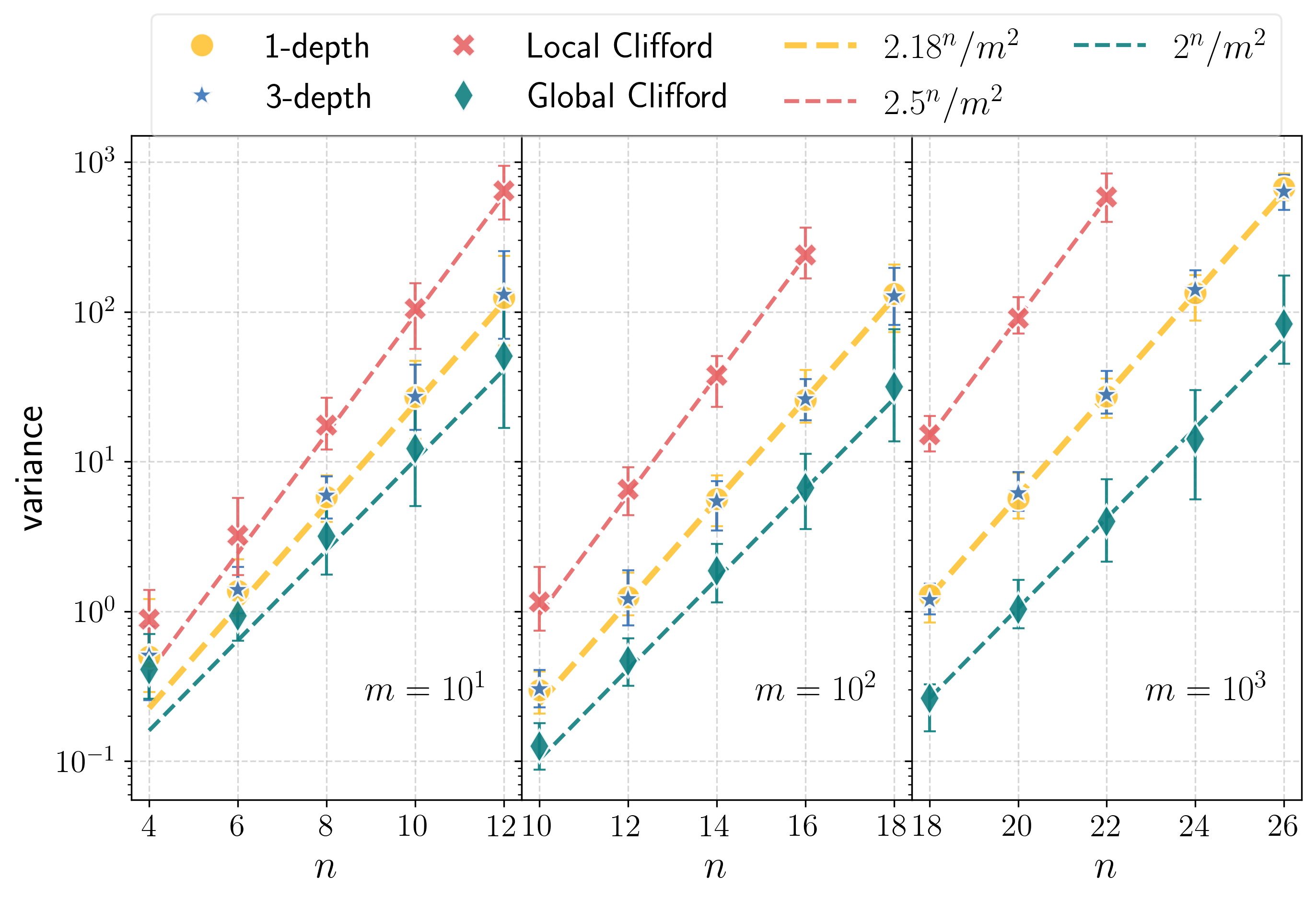}
\caption{
Numerical result of Haar random states. 
}
\label{fig: experiment of haar random states}
\end{figure}

\subsection{Fourth term of the variance}

\subsubsection{Haar random states}

Here, we show that the average of $\Var_{\cE}^{(4)}(\rho,\sigma)$ decreases exponentially with the qubit number $n$. 
We consider $\Cl_1^{\ox n}$ and $\cB_d$ ($d=1,3, 5, 7$) across systems ranging from $4$ to $10$ qubits in steps of $2$. 
For each $n$, we generate $10^2$ pairs of Haar random states $\{\ket{\psi_i}, \ket{\psi_i}\}_{i=1}^{100}$. 
For each pair, we sample $10^2$ unitaries from each ensemble and estimate their inner product using $m = 5\times 10^3$ measurement shots, such that the variance is approximately close to $\Var_{\cE}^{(4)}(\rho,\sigma)-1$. 
We compute the variance of these estimators for each ensemble, and the results are shown in Fig.~\ref{fig: experiment of 4-moment}.
We can find that $\Var_{\cE}^{(4)}(\rho,\sigma)-1$ decreases exponentially with the qubit number $n$, as shown in Lemmas~\ref{lem:average variance of local 2-design} and~\ref{lem:average variance of the brickwork}.

\begin{figure}[!htbp]
\centering
\includegraphics[width=0.9\linewidth]{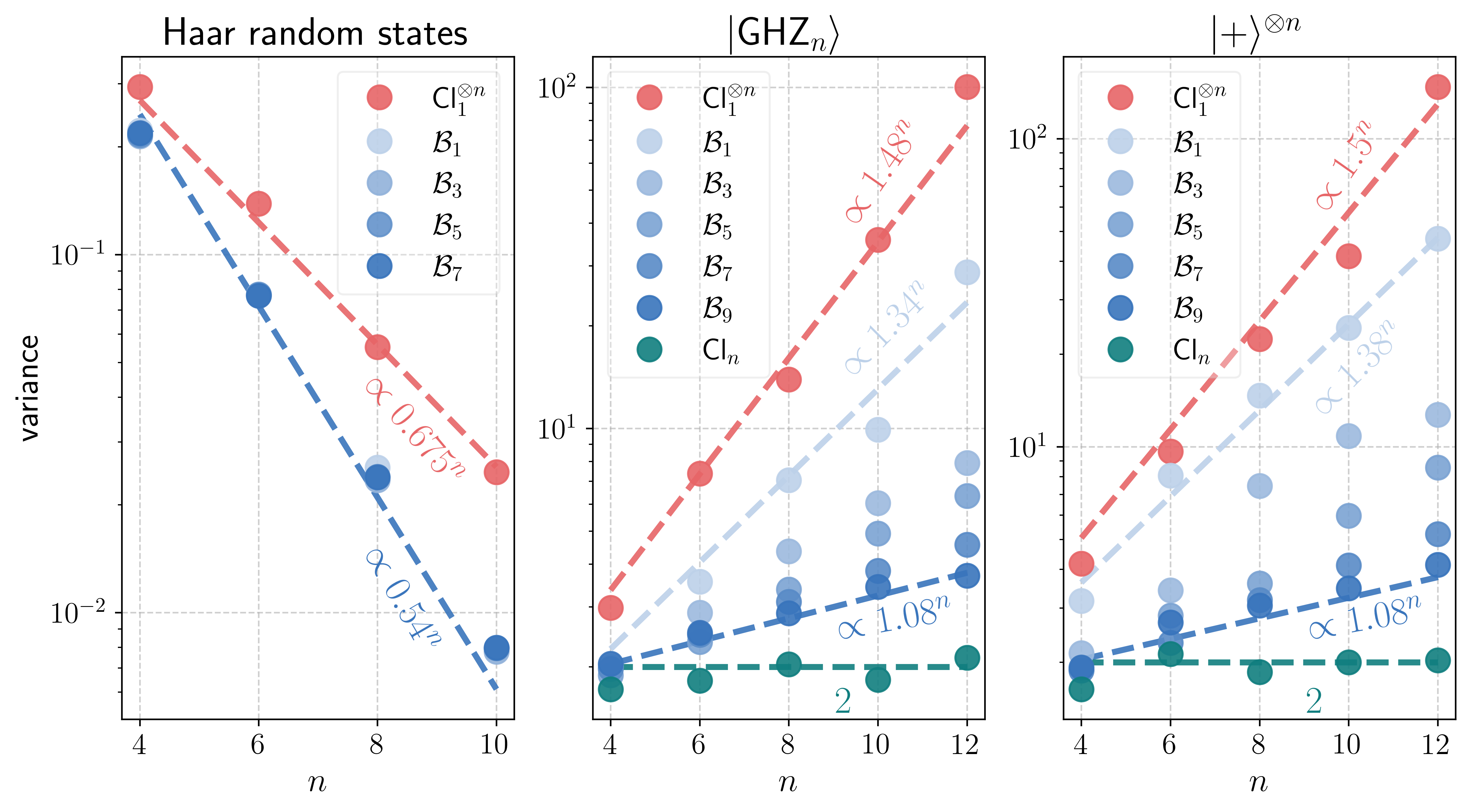}
\caption{
Numerical result of $\Var_{\cE}^{(4)}(\rho,\sigma)-1$. 
}
\label{fig: experiment of 4-moment}
\end{figure}

\subsubsection{Stabilizer states}

We further investigate the dependence of $\Var_{\cE}^{(4)}(\rho,\sigma)$ on the system size $n$, focusing on the cases where $\rho = \sigma$ is either the GHZ state $\ket{{\rm GHZ}_n}$ or the stabilizer product state $\ket{+}^{\ox n}$, with $\ket{+} = (\ket{0} + \ket{1})/\sqrt{2}$.
Specifically, we consider $\Cl_1^{\ox n}$, $\cB_d$ ($d=1,3, 5, 7, 9$), and $\Cl_n$ across systems ranging from $4$ to $12$ qubits in steps of $2$. 
Likewise, we generate $10^2$ pairs of states: $\{\ket{{\rm GHZ}_n}, \ket{{\rm GHZ}_n}\}$ and $\{\ket{+}^{\ox n}, \ket{+}^{\ox n}\}$. 
For each pair, we sample $10^2$ unitaries from each ensemble and estimate their inner product using $m = 5\times 10^3$ measurement shots. 
We compute the variance of these estimators for each ensemble, and the results are shown in Fig.~\ref{fig: experiment of 4-moment}.
We observe that for the local Clifford ensemble, the variance satisfies $\Var_{\Cl_1^{\ox n}}^{(4)}(\ketbra{+}^{\ox n}, \ketbra{+}^{\ox n}) - 1 = 1.5^n$, consistent with Theorem~\ref{the:local clifford variance}.
For the global Clifford ensemble, the variance remains constant with system size, satisfying $\Var_{\Cl_n}^{(4)}(\rho, \sigma) - 1 = 2$, as established in Theorem~\ref{the:global clifford variance}. For the brickwork ensembles $\cB_d$, we find an exponential scaling of the form $\Var_{\cB_d}^{(4)}(\rho, \sigma) - 1 \propto \alpha_d^n$, where the base $\alpha_d$ approaches 1 as the circuit depth $d$ increases.

\subsection{State-dependent Variance of Brickwork Ensemble}

Here, we investigate the influence of circuit depth on DIPE with the brickwork ensemble.
We focus on the quantity
\begin{align}
    \Upsilon_d(P) = \sum_{\ba} f_d^2(\ba, \bm{0}) h(\ba, P), 
\end{align}
which is defined in Lemma~\ref{lem:average variance of the brickwork}.
The value of $\Upsilon_d(P)$ determines the second term of the variance. 
As shown in Appendix~\ref{app:DIPE with random brickwork ensemble}, the function $h(\ba, P)$ depends only on the bitstring $x(P) \in \bZ_2^{n/2}$, defined by
\begin{align}
[x(P)]_{i} = \begin{cases}
    0, & [\gamma(P)]_{2i} \cdot [\gamma(P)]_{2i+i} = 0 \\
    1, & \text{otherwise}. 
\end{cases}
\end{align}
where $\gamma$ is defined in Eq.~\eqref{eq:signature of P}. 
Using the tensor network approach described in Appendix~\ref{app:DIPE with random brickwork ensemble}, we can compute $\Upsilon_d(P)$ efficiently. 
The results for $n = 6$ and depths $d = 1$ to $9$ are shown in Fig.~\ref{fig: experiment of brickwork}.
We observe that for nontrivial Pauli operators ($P \neq \1$), the value $\Upsilon_d(P)$ converges to $2^n$ as the depth increases.
This suggests that for all quantum states, the second term of the variance approaches its average behavior in the deep-circuit limit.

\begin{figure}[!htbp]
\centering
\includegraphics[width=0.7\linewidth]{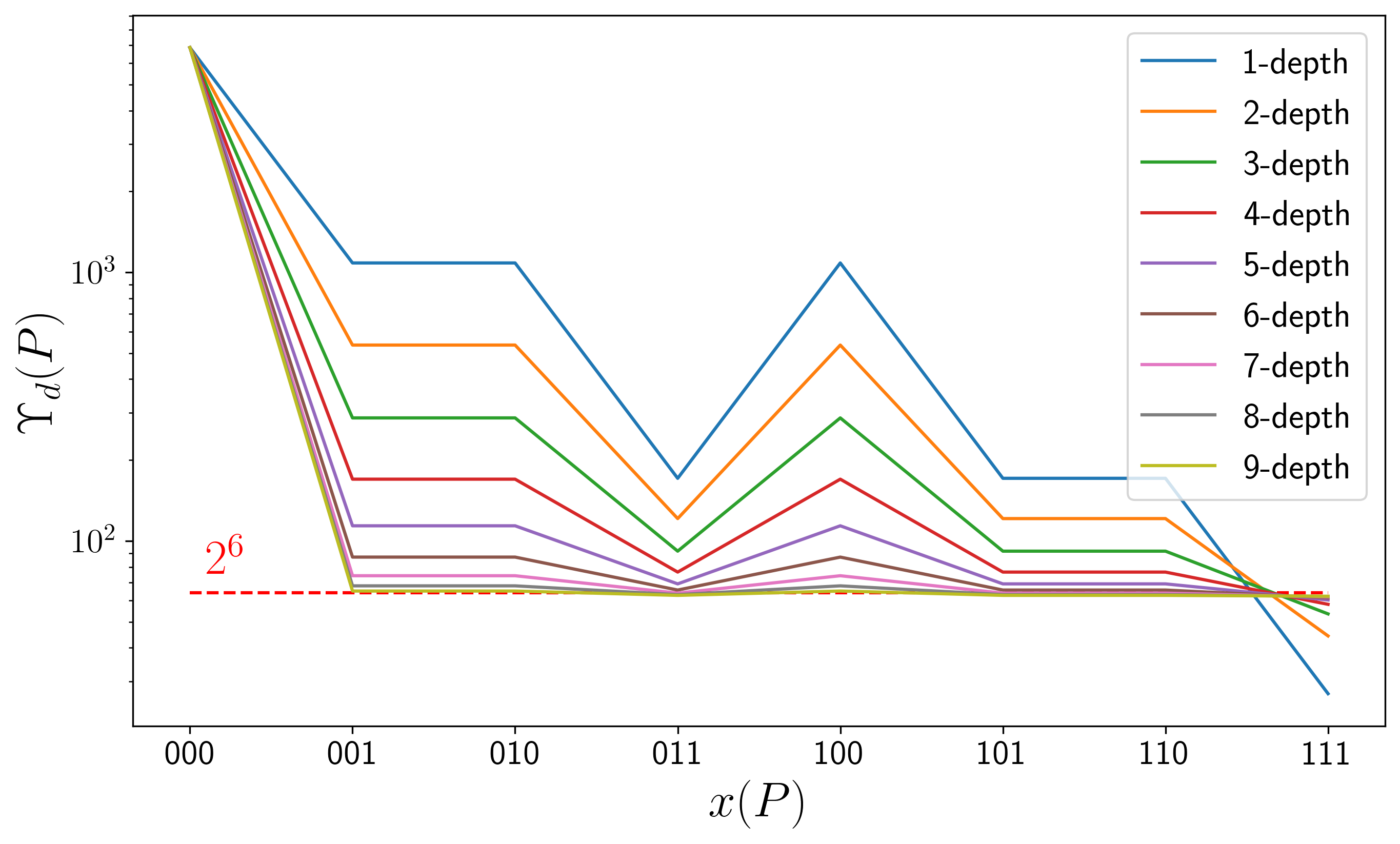}
\caption{
Numerical result of $\Upsilon_d(P)$. 
}
\label{fig: experiment of brickwork}
\end{figure}


\section{Useful Lemmas}
\label{app:useful lemmas}

In the following, we summarize the lemmas used in the main text and preceding appendices.
Note that several of these lemmas are standard tools in the context of randomized measurements.
Therefore, we only briefly review them here.

\subsection{Unitary Design}

\begin{lemma}[Appendix A of \cite{elben2020crossplatform}]
\label{lem:unitary 2-design expand}
If $\cE$ is a unitary $2$-design, we have 
\begin{align}
    \cM_{\cE}^{(2)}(A) 
    = \frac{\tr[A] - \tr\left[(\bigotimes\bS) A\right]/2^n}{4^n-1} \1 + 
    \frac{\tr[(\bigotimes\bS) A] - \tr[A]/2^n}{4^n-1} \bigotimes\bS, 
\end{align}
where $\bS$ is the SWAP operator. 
\end{lemma}

\begin{lemma}[Lemma 22 of \cite{anshu2022distributed}]
\label{lem:properties of unitary design}
Let $A,B,C,D$ be Hermitian matrices. 
If $\cE$ is a unitary $2$-design, we have 
\begin{align}
    \bE_{U\sim\cE} \bra{\ba}U A U^\dagger\ket{\ba} \bra{\ba}U B U^\dagger\ket{\ba} = \frac{1}{2^n(2^n+1)} \left(\tr[A]\tr[B] + \tr[AB]\right), \quad \forall\;\ba\in\bZ_2^n.
\end{align}
If $\cE$ is a unitary $3$-design, we have 
\begin{align}
    \bE_{U\sim\cE} \bra{\ba}U A U^\dagger\ket{\ba} \bra{\ba}U B U^\dagger\ket{\ba} \bra{\ba}U C U^\dagger\ket{\ba} \notag 
    = &\frac{1}{2^n(2^n+1)(2^n+2)} \left(\tr[A]\tr[B]\tr[C] + 
    \tr[AB]\tr[C] \right. \\ 
    &\left. + \tr[A]\tr[BC] + \tr[B]\tr[AC] + \tr[ABC] + \tr[ACB]\right). 
\end{align}
If $\cE$ is a unitary $4$-design, we have 
\begin{align}
    &\bE_{U\sim\cE} \bra{\ba}U A U^\dagger\ket{\ba} \bra{\ba}U B U^\dagger\ket{\ba} \bra{\ba}U C U^\dagger\ket{\ba} \bra{\ba}U D U^\dagger\ket{\ba} 
    = \frac{1}{2^n(2^n+1)(2^n+2)(2^n+3)} \sum_{\pi\in\cS_4} \tr[\bP_{\pi} (A\ox B\ox C\ox D)], 
\end{align}
where $\bP_{\pi}$ is the permutation operator, defined as 
\begin{align}
    \bP_\pi = \sum_{\ba_1,\ba_2,\cdots,\ba_k} \ket{\ba_{\pi^{-1}(1)} \ba_{\pi^{-1}(2)} \cdots \ba_{\pi^{-1}(k)}}\!\bra{\ba_1\ba_2\cdots\ba_k}, 
\end{align} 
for $\pi\in\cS_k$. 
\end{lemma}

\begin{lemma}
\label{lem:property of unitary 4-design}
If $\cE$ is a unitary $4$-design, we have 
\begin{align}
    \tr\left[\cM_{\cE}^{(4)}(\Lambda_1)(\rho\ox\sigma)^{\ox 2}\right]
    = \frac{(1+\tr[\rho\sigma])^2}{2^n(2^n+1)} + \cO(2^{-3n})
\end{align}
where $\Lambda_1 = \sum_{\ba,\bb}\ketbra{\ba\ba\bb\bb}$. 
\end{lemma}
\begin{proof}
This lemma is proven based on the results of~\cite{anshu2022distributed}. 
L.H.S. can be rewritten as  
\begin{align}
    \tr\left[\cM_{\cE}^{(4)}(\Lambda_1)(\rho\ox\sigma)^{\ox 2}\right]
    &= \sum_{\ba} \bE_{U\sim\cE} (\bra{\ba}U\rho U^\dagger\ket{\ba} \bra{\ba}U\sigma U^\dagger\ket{\ba})^2 + 
    \sum_{\ba\neq \bb} \bE_{U\sim\cE} (\bra{\ba}U\rho U^\dagger\ket{\ba} \bra{\bb}U\sigma U^\dagger\ket{\bb})^2 \\ 
    &= 2^n\cO(2^{-4n}) + \frac{(1+\tr[\rho\sigma])^2}{2^n(2^n+1)} + \cO(2^{-4n}) \\
    &= \frac{(1+\tr[\rho\sigma])^2}{2^n(2^n+1)} + \cO(2^{-3n})
\end{align}
where the second line uses Lemma~\ref{lem:properties of unitary design} and~\cite[Eq.~(194)]{anshu2022distributed}. 
\end{proof}

\subsection{Clifford ensemble}

\begin{lemma}
\label{lem:4-moment of global clifford}
Suppose that $\Cl_n$ is the $n$-qubit global Clifford ensemble, 
for states $\rho, \sigma$, we have 
\begin{align}
    \tr\left[\cM_{\Cl_n}^{(4)}(\Lambda_2)(\rho\ox\sigma)^{\ox 2}\right]
    &= \frac{1}{(2^n+1)(2^n+2)} \left[(1+\tr[\rho\sigma])^2 + \frac{1}{2^n}\left(\norm{\Xi_{\rho,\sigma}}{2}^2 + 
    \tilde{\Xi}_{\rho,\sigma} \cdot \Xi_{\rho,\sigma}\right)\right] \\
    &\leq \frac{1}{(2^n+1)(2^n+2)} \left[(1+\tr[\rho\sigma])^2 + \frac{1}{2^{n-1}} \norm{\Xi_{\rho,\sigma}}{2}^2\right] 
\end{align}
where 
\begin{align}
    \norm{\Xi_{\rho,\sigma}}{2}^2 
    &:= \sum_{P\in\cP_n} \tr^2[P\rho] \tr^2[P\sigma], \\
    \tilde{\Xi}_{\rho,\sigma} \cdot \Xi_{\rho,\sigma}
    &:= \sum_{P\in\cP_n} \tr[\rho P\sigma P] \tr[P\rho] \tr[P\sigma]. 
\end{align}
\end{lemma}

\begin{proof}[Proof of Lemma~\ref{lem:4-moment of global clifford}]
This lemma is proven based on the results of~\cite{chen2024nonstabilizerness}.
With~\cite[Eq.~(G66)]{chen2024nonstabilizerness} and~\cite[Lemma 17]{chen2024nonstabilizerness}, we can decompose $\cM_{\Cl_n}^{(4)}(\Lambda_1)$ as 
\begin{align}
    \cM_{\Cl_n}^{(4)}(\Lambda_1)
    = \frac{1}{(2^n+1)(2^n+2)} \left(\cR_1 + \cR_4\right),  
\end{align}
where 
\begin{align}
    \cR_1 := \bP_{(e)} + \bP_{(12)} + \bP_{(34)} + \bP_{(12)(34)}, \quad
    \cR_4 := R_{T_4} + \bP_{(12)}R_{T_4}, \quad
    R_{T_4} := \frac{1}{2^n}\sum_{P\in\cP_n} P^{\ox 4}, 
\end{align}
as defined in \cite[Eq.~(F38)]{chen2024nonstabilizerness}. 
Then, we have 
\begin{align}
    \tr[\cR_1(\rho\ox\sigma)^{\ox 2}] 
    = 1 + 2\tr[\rho\sigma] + \tr^2[\rho\sigma] = (1+\tr[\rho\sigma])^2
\end{align}
and 
\begin{align}
    \tr[\cR_4(\rho\ox\sigma)^{\ox 2}]
    &= \frac{1}{2^n}\left[\sum_{P\in\cP_n} \tr^2[P\rho] \tr^2[P\sigma] + 
    \sum_{P\in\cP_n} \tr[\rho P\sigma P] \tr[P\rho] \tr[P\sigma]\right] \\
    &= \frac{1}{2^n}\left[\norm{\Xi_{\rho,\sigma}}{2}^2 + 
    \tilde{\Xi}_{\rho,\sigma} \cdot \Xi_{\rho,\sigma}\right] \\
    &\leq \frac{1}{2^{n-1}} \norm{\Xi_{\rho,\sigma}}{2}^2, 
\end{align}
where the last line use~\cite[Lemma~3]{chen2024nonstabilizerness}. 
\end{proof}

\subsection{Haar Random States}

\begin{lemma}[Lemma 1 of~\cite{anshu2022distributed}]
\label{lem:haar random states}
Given a Haar random state $\ket{\psi}$ in $\cH_n$, we have the following results
\begin{align}
\bE_{\psi} \ketbra{\psi} 
&= \frac{1}{2^n} \1, \\
\bE_{\psi} \ketbra{\psi}^{\ox 2} 
&= \frac{1}{2^n(2^n+1)}\left(\1 + \bigotimes \bS\right), \\
\bE_{\psi} \ketbra{\psi}^{\ox k} 
&= \frac{1}{2^n(2^n+1)\cdots(2^n+k-1)}\sum_{\pi\in\cS_k} \bP_\pi. 
\end{align}
\end{lemma}
\end{document}

%% file: output.bbl
%